\pdfoutput=1
\RequirePackage{ifpdf}
\ifpdf % We are running pdfTeX in pdf mode
\documentclass[pdftex]{sigma}
\else
\documentclass{sigma}
\fi

\numberwithin{equation}{section}

\newtheorem{Theorem}{Theorem}[section]
\newtheorem{Lemma}[Theorem]{Lemma}
\newtheorem{Proposition}[Theorem]{Proposition}
 { \theoremstyle{definition}
\newtheorem{Definition}[Theorem]{Definition}}

\begin{document}

%\allowdisplaybreaks

\newcommand{\arXivNumber}{1702.01227}

\renewcommand{\PaperNumber}{035}

\FirstPageHeading

\ShortArticleName{Liouville Correspondences between Integrable Hierarchies}

\ArticleName{Liouville Correspondences\\ between Integrable Hierarchies}

\Author{Jing KANG~$^{\dag}$, Xiaochuan LIU~$^{\dag}$, Peter J.~OLVER~$^{\ddag}$ and Changzheng QU~$^{\S}$}

\AuthorNameForHeading{J.~Kang, X.C.~Liu, P.J.~Olver and C.Z.~Qu}

\Address{$^{\dag}$~Center for Nonlinear Studies and School of Mathematics, Northwest University,\\
\hphantom{$^{\dag}$}~Xi'an 710069, P.R.~China}
\EmailD{\href{mailto:jingkang@nwu.edu.cn}{jingkang@nwu.edu.cn}, \href{mailto:liuxc@nwu.edu.cn}{liuxc@nwu.edu.cn}}

\Address{$^{\ddag}$~School of Mathematics, University of Minnesota, Minneapolis, MN 55455, USA}
\EmailD{\href{mailto:olver@umn.edu}{olver@umn.edu}}
\URLaddressD{\url{http://www.math.umn.edu/~olver/}}

\Address{$^{\S}$~Center for Nonlinear Studies and Department of Mathematics, Ningbo University,\\
\hphantom{$^{\S}$}~Ningbo 315211, P.R.~China}
\EmailD{\href{mailto:quchangzheng@nbu.edu.cn}{quchangzheng@nbu.edu.cn}}

\ArticleDates{Received February 07, 2017, in f\/inal form May 22, 2017; Published online May 28, 2017}

\Abstract{In this paper, we study explicit correspondences between the integrable Novi\-kov and Sawada--Kotera hierarchies, and between the Degasperis--Procesi and Kaup--Kuper\-shmidt hierarchies. We show how a pair of Liouville transformations between the isospectral problems of the Novikov and Sawada--Kotera equations, and the isospectral problems of the Degasperis--Procesi and Kaup--Kupershmidt equations relate the corresponding hierarchies, in both positive and negative directions, as well as their associated conservation laws. Combining these results with the Miura transformation relating the Sawada--Kotera and Kaup--Kupershmidt equations, we further construct an implicit relationship which associates the Novikov and Degasperis--Procesi equations.}


\Classification{37K05; 37K10}

\section{Introduction}

This paper is devoted to studying Liouville correspondences between certain integrable hierarchies and their interrelationships. One pair consists of the Novikov and Sawada--Kotera (SK) hierarchies, which are initiated respectively from the Novikov equation \cite{hw2, nov}
\begin{gather}\label{novikov}
m_t=3uu_xm+u^2 m_x, \qquad m=u-u_{xx},
\end{gather}
and the SK equation \cite{cdg,sk}
\begin{gather}\label{SK}
Q_{\tau}+Q_{yyyyy}+5QQ_{yyy}+5Q_yQ_{yy}+5 Q^2 Q_y=0.
\end{gather}
The second pair of integrable hierarchies is initiated by the Degasperis--Procesi (DP) equa\-tion~\cite{dhh, dp}
\begin{gather}\label{DP}
m_t=3u_xm+u m_x, \qquad m=u-u_{xx},
\end{gather}
and the Kaup--Kupershmidt (KK) equation \cite{kau, kup}
\begin{gather}\label{KK}
P_{\tau}+P_{yyyyy}+10 PP_{yyy}+25 P_yP_{yy}+80P^2 P_y=0.
\end{gather}
Furthermore, combining these explicit correspondences with the known Miura transformation linking the SK and KK equations \cite{fg}, we derive a nontrivial underlying correspondence between the Novikov equation \eqref{novikov} and the DP equation~\eqref{DP}.

The DP equation \eqref{DP} was derived by Degasperis and Procesi \cite{dp} as a result of the asymptotic integrability method for classifying (a class of) third-order nonlinear dispersive evolution equations. It was subsequently shown that the DP equation is integrable with a Lax pair invol\-ving a $3\times 3$ isospectral problem as well as a bi-Hamiltonian structure \cite{dhh}. Furthermore, the Lax representation can be written in the form of a matrix spectral problem of Zakharov--Shabat (ZS) form. The associated inverse scattering transform and the dressing method can be applied to construct smooth soliton solutions for the DP equation \cite{ci, cil, mik}. Physically, the DP equation provides a model describing the propagation of shallow water waves \cite{cl, dgh}. It admits peaked solitons (peakons) \cite{dhh} as well as multi-peakon solutions which recover the soliton interaction dynamics \cite{dhh, ls} and exhibits a particular shock peakon structure \cite{lun}. The integrability, well-posedness, wave breaking phenomenon and stability of peakons for the DP equation have been studied extensively; see \cite{ck2, hw1, ll, ly} and references therein. The Novikov equation \eqref{novikov} with cubic nonlinear terms was discovered as a consequence of the symmetry classif\/ication of nonlocal partial dif\/ferential equations involving both cubic and quadratic nonlinearities \cite{nov}. A Lax pair formulation based on a $3\times 3$ isospectral problem and the associated bi-Hamiltonian structure were established in \cite{hw2}. It was also shown that the Novikov equation possesses peaked solitons and multi-peakon solutions \cite{hls, hw2}. The well-posedness, wave breaking and blow-up phenomena, as well as stability of peakons for the Novikov equation have been studied in a number of papers, including \cite{hh, llq1, tig}.

The SK equation \eqref{SK} and the KK equation \eqref{KK} are two typical f\/ifth-order integrable equations \cite{cdg, kau, kup, sk}. Their integrability can be verif\/ied from several dif\/ferent standpoints: for instance, they both possess $3\times 3$ isospectral problems and bi-Hamiltonian structures, enjoy the Painlev\'{e} property, admit multi-soliton solutions, etc. Like the DP equation, the KK equation also admits a Lax operator of ZS form, which can be applied to construct solitons of the KK hie\-rarchy, \cite{bcg, ger, mik, val}. In Section~\ref{section2}, we will see that both the Novikov and SK equations support Lax operators of ZS form. Geometrically, the SK equation arises naturally from an integrable planar curve f\/low in af\/f\/ine geometry \cite{cq,olv3}, while the KK equation comes from an integrable planar curve f\/low in projective geometry~\cite{lq,mus}. Interestingly, both equations are related to the so-called Fordy--Gibbons--Jimbo--Miwa equation via certain Miura transformations~\cite{fg}.

Like the Camassa--Holm (CH) equation \cite{ch, chh, ff} and the modif\/ied Camassa--Holm (mCH) equation \cite{fuc, or}, the Novikov and DP equations exhibit nonlinear dispersion. Recent years have seen a proliferation of papers for the CH and mCH equations studying their integrable properties, geometric formulations, well-posedness for solutions of the Cauchy problem, and the stability of peaked solitary waves solutions; see \cite{achm, bss, cht, ce-1, cgi, cl,cm,cs,gloq, hi, joh, kou, lo, llq2, mck-1}. Interestingly, these equations support a notable variety of non-smooth soliton-like solutions and can model the phenomenon of wave breaking.

The method of tri-Hamiltonian duality was developed \cite{for, fuc, or} to systematically derive additional nonlinear dispersive integrable systems. This method begins with the basic observation that most classical integrable soliton equations that possess a bi-Hamiltonian structure, actually support a compatible triple of Hamiltonian structures using a particular scaling argument, leading to a systematic algorithm \cite{or}, to construct their dual nonlinear dispersive integrable systems. In particular, the CH and mCH equations appear as the duals to, respectively, the KdV and mKdV equations.

In view of this duality, it is of interest to establish relationships between the full integrable hierarchy and the corresponding dual integrable hierarchy. In~\cite{len1} and~\cite{mck}, the correspondence between the CH hierarchy and the KdV hierarchy is established through a Liouville transformation, the key point being that the two Hamiltonian operators of the CH hierarchy can be obtained directly from those of the KdV hierarchy. This argument does not work for the mCH and mKdV hierarchies, whose correspondence through a Liouville transformation is based upon a relationship between the corresponding recursion operators and some subtle identities relating the respective Hamiltonian operators~\cite{kloq}. It demonstrates that the positive f\/low and negative f\/low of the mCH hierarchy are generated by the negative f\/low and positive f\/low of the mKdV hierarchy, respectively. The correspondences between the Hamiltonian conservation laws for the CH (mCH) hierarchy and KdV (mKdV) hierarchy have also been derived~\cite{kloq, len1}.

The goal of this paper is to study the similar Liouville correspondences between the f\/lows and Hamiltonian conservation laws in both the Novikov and SK hierarchies, as well as the DP and KK hierarchies. Furthermore, an underlying correspondence between the Novikov equation~\eqref{novikov} and the DP equation~\eqref{DP} is also constructed. Our motivations are three-fold. First, it was shown that the Novikov equation is related to the f\/irst negative f\/low of the SK hierarchy~\cite{hw2}, while the DP equation is related to the f\/irst negative f\/low of the KK hierarchy~\cite{dhh}. Second, the CH and mCH hierarchies are related, respectively, to the KdV and mKdV hierarchies through Liouville transformations relating their isospectral problems. Third, the SK equation is related to the KK equation by a Miura transformation~\cite{fg}, and there exists a transformation found in~\cite{kloq} which maps the mCH equation to the CH equation.

However, in the Novikov-SK and DP-KK settings, due to the non-standard bi-Hamiltonian structures \cite{fo}, we neither have the dual relationship, as in both CH-KdV and mCH-mKdV settings, nor the subtle relationship between their Hamiltonian operators, as in the CH-KdV setting \cite{len1, mck}, nor between their recursion operators, as in the mCH-mKdV setting \cite{kloq}. On the other hand, given that the Novikov and DP equations are both third-order nonlinear equations, while the SK and KK equations are of f\/ifth order, it seems dif\/f\/icult to establish any relationship between the Novikov or DP equations with the f\/lows in the negative direction of the SK hierarchy or the KK hierarchy. Nevertheless, based on the Liouville transformation between the isospectral problems of the Novikov and SK hierarchies, as well as the DP and KK hierarchies, we are able to establish certain nontrivial identities which reveal the underlying relationship between the recursion operator of the Novikov (DP) hierarchy and the adjoint operator of the recursion operator for the SK (KK) hierarchy. Using these operator identities, we are able to prescribe a~Liouville correspondence between the f\/lows involved in the Novikov-SK hierarchies and DP-KK hierarchies.

It is worth noting that, in the Novikov-SK setting, in order to establish the explicit relationship between the f\/lows in the positive Novikov hierarchy and the f\/lows in the negative SK hierarchy, we make use of a novel factorization of the recursion operator of the SK equation to identify the equations transformed from the positive f\/lows in the Novikov hierarchy as the corresponding negative f\/lows in the SK hierarchy exactly. The factorization is based on the following nontrivial operator identity for the recursion operator of the SK equation \cite{cq}:
\begin{gather*}
\bar{\mathcal{R}}=-\big( \partial_y^3+2 Q \partial_y+2 \partial_y Q\big) \big(2\partial_y^3+2 \partial_y^2 Q \partial_y^{-1}+2 \partial_y^{-1} Q\partial_y^2+Q^2 \partial_y^{-1}+\partial_y^{-1} Q^2\big)\\
\hphantom{\bar{\mathcal{R}}}{} =-2\big( \partial_y^4+5Q \partial_y^2+4Q_y \partial_y+Q_{yy}+4Q^2+2Q_y \partial_y^{-1}Q\big) \big( \partial_y^2+Q+Q_y\partial_y^{-1}\big).
\end{gather*}

Since conservation laws play a key role in the study of well-posedness of solutions, stability of solitons, and wave-breaking phenomena, another topic of this paper is to establish relationships between the Hamiltonian conservation laws for the Novikov and SK hierarchies, and the DP and KK hierarchies. These rely on some new identities for Hamiltonian conservation laws related by the Liouville transformations and certain known results.

This section is concluded by outlining the rest of the paper. In Section~\ref{section2}, we f\/irst present the Liouville transformation relating the isospectral problems of the Novikov hierarchy and the SK hierarchy in Section~\ref{section2.1}. Next in Section~\ref{section2.2}, several operator identities are combined with the Liouville transformation to establish the one-to-one correspondence between the f\/lows in the Novikov and SK hierarchies. It is proved in Section~\ref{section2.3} that the Liouville transformation establishes the correspondence between the series of Hamiltonian conservation laws of the Novikov equation and the SK equation. The Liouville correspondence between the DP hierarchy and the KK hierarchy, and the relationship of their conservation laws will be studied in Section~\ref{section3}. In Section~\ref{section4}, we obtain a nontrivial relationship between the Novikov equation \eqref{novikov} and the DP equation \eqref{DP} by exploiting the Miura transformations relating the SK equation~\eqref{SK} and the KK equation~\eqref{KK} and the results in previous Sections.

\section{The correspondence between the Novikov and SK hierarchies}\label{section2}

\subsection[A Liouville transformation between the isospectral problems of the Novikov and SK hierarchies]{A Liouville transformation between the isospectral problems\\ of the Novikov and SK hierarchies}\label{section2.1}

In this section, we f\/irst obtain the Liouville transformation relating the Novikov and SK hierarchies. In accordance with standard terminology, a Liouville transformation is def\/ined by a~change of variables which maps one spectral problem to another~\cite{mil,folv}. If the transformation does not af\/fect the independent variables, it is referred to as a Miura transformation.

The Novikov equation
\begin{gather}\label{nov}
m_t=u^2 m_x+3uu_xm, \qquad m=u-u_{xx},
\end{gather}
can be expressed as the compatibility condition for the linear system \cite{hw2} consisting of
\begin{gather}\label{iso-nov}
\mathbf{\Psi}_x =\begin{pmatrix} 0 &\lambda m & 1\\
 0 &0 & \lambda m\\
 1 &0 & 0
 \end{pmatrix} \mathbf{\Psi},\qquad
\mathbf{\Psi}=\begin{pmatrix}\psi_1 \\ \psi_2 \\ \psi_3\end{pmatrix},
\end{gather}
and
\begin{gather*}%\label{iso-novt}
\mathbf{\Psi}_t
=\begin{pmatrix} \frac{1}{3}\lambda^{-2}-uu_x &\lambda^{-1}u_x-\lambda u^2 m & u_x^2\\
 \lambda^{-1}u & -\frac{2}{3}\lambda^{-2} & -\lambda^{-1}u_x-\lambda u^2 m\\
 -u^2 & \lambda^{-1} u & \frac{1}{3}\lambda^{-2}+uu_x
 \end{pmatrix}\mathbf{\Psi}.
\end{gather*}
Note that equation \eqref{iso-nov} is reduced to a scalar equation by setting $\Psi=\psi_2$, namely
\begin{gather}\label{iso-nov1}
\Psi_{xxx} =2m^{-1}m_x\Psi_{xx}+\big( m^{-1} m_{xx}-2m^{-2}m_x^2+1\big) \Psi_x+\lambda^2 m^2 \Psi.
\end{gather}
It was proved in \cite{hw2} that by the reciprocal transformation
\begin{gather}\label{tran1-nov}
\mathrm{d} y=m^{\frac{2}{3}} \mathrm{d}x+m^{\frac{2}{3}} u^2 \mathrm{d}t, \qquad \mathrm{d} \tau=\mathrm{d} t,
\end{gather}
the isospectral problem (\ref{iso-nov1}) is converted into
\begin{gather}\label{iso-sk}
\Phi_{yyy}+Q \Phi_{y}= \mu \Phi,
\end{gather}
with
\begin{gather}\label{Qm1}
\Phi=\Psi,\qquad \mu=\lambda^2, \qquad Q=\frac{4}{9}m^{-\frac{10}{3}}m_x^2-\frac{1}{3}m^{-\frac{7}{3}}m_{xx}-m^{-\frac{4}{3}},
\end{gather}
which is a third-order spectral problem for the SK equation. Note that the isospectral problems for the Novikov equation and the SK equation can also be written as the Zakharov--Shabat (ZS) form
\begin{gather*}%\label{iso-zs}
\tilde{\Psi}_y+\big(Q_l -\tilde{\lambda} J\big) \tilde{\Psi}=0,
\end{gather*}
where $\tilde{\lambda}=\lambda^{\frac 23}$, while
\begin{gather*}
Q_l =\begin{pmatrix} -g^{-1}g_y &0 & 0\\
 0 &g^{-1}g_y+h^{-1}h_y & 0\\
 0 &0 & -h^{-1}h_y
 \end{pmatrix},
\end{gather*}
is a diagonal $\mathfrak{sl}(3)$ matrix. The functions here, $g(y,t)$, $h(y,t)$ satisfy the system
\begin{gather*}
g_{yy}+h^{-1}h_yg_y+hg=0, \qquad
h^{-1}h_{yy}-2h^{-1}h_y^2+h=-\frac{1}{3} m^{-1}m_{yy}+\frac{2}{9} m^{-2}m_y^2-m^{-\frac 43},
\end{gather*}
in the case of the Novikov equation, and
\begin{gather*}
g_{yy}+h^{-1}h_yg_y+hg=0, \qquad h^{-1}h_{yy}-2h^{-1}h_y^2+h=Q,
\end{gather*}
for the SK equation.

Moreover, using \eqref{tran1-nov},
\begin{gather*}
\partial_t=\partial_\tau+m^{\frac{2}{3}} u^2\partial_y,
\end{gather*}
the $t$ evolution of $\Psi=\psi_2$ in \eqref{iso-nov} is transformed into
\begin{gather}\label{phitau1}
\Phi_\tau
-\frac{1}{\mu} ( V \Phi_{yy}-V_y\Phi_y ) +\frac{2}{3\mu}\Phi=0, \qquad \mathrm{with}\qquad V=um^{\frac{1}{3}}.
\end{gather}
Notice that (\ref{phitau1}) is equivalent to
\begin{gather}\label{phitau}
\Phi_\tau
+\frac{1}{3\mu} ( W \Phi_{yy}-W_y\Phi_y )=0
\end{gather}
after gauging $\Phi$ by a factor of $e^{2\tau/(3\mu)}$ and setting $W=-3V$. Indeed, the linear system~(\ref{iso-sk}) and~(\ref{phitau}) provides the Lax pair for the f\/irst negative f\/low in the SK hierarchy~\cite{gp} (see also~\cite{hw1, hw2}), and the compatibility condition $\Phi_{yyy\tau}=\Phi_{\tau yyy}$ yields
\begin{gather}\label{sk-1'}
Q_\tau=W_y, \qquad W_{yy}+Q W=T,\qquad T_y=0.
\end{gather}
Therefore, we conclude that there exists a Liouville correspondence between the Novikov equation \eqref{nov} and the f\/irst negative f\/low \eqref{sk-1'} of the SK hierarchy, where their corresponding Lax pairs are related by the transformations (\ref{tran1-nov}) and~(\ref{Qm1}).

In light of this, we are led to generalize the Liouville correspondence between the Novikov equation and the f\/irst negative f\/low of the SK hierarchy to their entire hierarchies, establishing the correspondence between the f\/lows of the Novikov and SK hierarchies. Motivated by (\ref{tran1-nov}) and (\ref{Qm1}), we pursue this study by utilizing the Liouville transformation
\begin{gather}\label{recixy-nov}
y=\int^x m^{\frac{2}{3}}(t,\xi)\mathrm{d}\xi, \qquad \tau=t, \qquad
Q=\frac{4}{9}m^{-\frac{10}{3}}m_x^2-\frac{1}{3}m^{-\frac{7}{3}}m_{xx}-m^{-\frac{4}{3}}.
\end{gather}
Note that the f\/irst expression in \eqref{recixy-nov} has the form of a reciprocal transformation \cite{rs}.

\subsection{The correspondence between the Novikov and SK hierarchies}\label{section2.2}

Let us now study the correspondence between the Novikov hierarchy and the SK hierarchy. First of all, the Novikov equation \eqref{nov} can be written in bi-Hamiltonian form \cite{hw2}
\begin{gather}\label{nov-bi}
m_t=K_1=\mathcal{K} \frac{\delta \mathcal{H}_0}{\delta m}=\mathcal{J} \frac{\delta \mathcal{H}_1}{\delta m},\qquad m=u-u_{xx},
\end{gather}
where
\begin{gather}\label{nov-kj}
\mathcal{K}=\frac{1}{2}m^{\frac{1}{3}} \partial_x m^{\frac{2}{3}} \big(4\partial_x-\partial_x^3\big)^{-1} m^{\frac{2}{3}} \partial_x m^{\frac{1}{3}}, \qquad
\mathcal{J}=\big(1-\partial_x^2\big) m^{-1} \partial_x m^{-1} \big(1-\partial_x^2\big)
\end{gather}
are the compatible Hamiltonian operators. The corresponding Hamiltonian functionals are given by
\begin{gather*}
\mathcal{H}_0= 9\int \big( u^2+u_x^2\big) \mathrm{d}x
\end{gather*}
and
\begin{gather*}%\label{nov-H1}
\mathcal{H}_1= \frac{1}{6}\int um\partial_x^{-1}m\big(1-\partial_x^2\big)^{-1}\big(u^2m_x+3uu_xm\big) \mathrm{d}x.
\end{gather*}
Moreover, since
\begin{gather*}
\partial_x^{-1}(mu_t)=\int_{-\infty}^x(u-u_{xx})u_t \mathrm{d}x=-(u_xu_t-uu_{xt})(t,x)+\int_{-\infty}^xu(u_t-u_{xxt}) \mathrm{d}x\\
\hphantom{\partial_x^{-1}(mu_t)}{} =-(u_xu_t-uu_{xt})(t,x)+\int_{-\infty}^xu\big(u^2m_x+3uu_x m\big) \mathrm{d}x\\
\hphantom{\partial_x^{-1}(mu_t)}{} =\big(u u_{xt}-u_x u_t+u^3 m\big)(t,x),
\end{gather*}
using the Novikov equation \eqref{nov},
\begin{gather*}
\mathcal{H}_1 = \frac 16 \int_{\mathbb R} um\partial_x^{-1}(mu_t)\mathrm{d}x=\frac 16 \int_{\mathbb R} u m\big(uu_{xt}-u_xu_t+u^3m\big)\mathrm{d}x\\
\hphantom{\mathcal{H}_1}{} = \frac 16 \int_{\mathbb R} u^2 mu_{xt}\mathrm{d}x-\frac 16 \int_{\mathbb R}uu_xmu_tdx+\frac 16 \int_{\mathbb R}u^4 m^2 \mathrm{d}x\\
\hphantom{\mathcal{H}_1}{} = \frac 16 \int_{\mathbb R} \big({-}3uu_xm-u^2 m_x\big)u_t \mathrm{d}x+\frac 16 \int_{\mathbb R}u^4 m^2 \mathrm{d}x ,
\end{gather*}
which implies that $\mathcal{H}_1$ can be written in the following local form in terms of $u$ and $m$:
\begin{gather*}%\label{nov-H2}
\mathcal{H}_1= \frac{1}{6}\int \big( u^4m^2-m_tu_t \big) \mathrm{d}x.
\end{gather*}

According to Magri's theorem \cite{mag, olv1, olv2}, an integrable bi-Hamiltonian equation with two compatible Hamiltonian operators $\mathcal{K}$ and $\mathcal{J}$ belongs to an inf\/inite hierarchy
\begin{gather}\label{ph}
m_t=K_n=\mathcal{K} \frac{\delta \mathcal{H}_{n-1}}{\delta m}=\mathcal{J} \frac{\delta \mathcal{H}_n}{\delta m},\qquad n\in \mathbb{Z}
\end{gather}
of higher-order bi-Hamiltonian systems, in both the positive and negative directions, where $\mathcal{H}_{n}$, $n\in \mathbb{Z}$ are all conserved functionals common to all members of the hierarchy.

The Novikov equation \eqref{nov-bi} serves as the f\/irst member in the positive direction of \eqref{ph}. As for the negative direction, observe that
\begin{gather*}
K_0=\mathcal{J} \frac{\delta \mathcal{H}_0}{\delta m}=0,
\end{gather*}
and the Hamiltonian operator $\mathcal{K}$ admits the Casimir functional
\begin{gather*}
\mathcal{H}_C= \frac{9}{2} \int m^{\frac{2}{3}} \mathrm{d}x\qquad \hbox{with }\qquad \frac{\delta \mathcal{H}_C}{\delta m}=3 m^{-\frac{1}{3}}.
\end{gather*}
Therefore, we conclude that the negative f\/lows of the Novikov hierarchy are generated from the Casimir equation
\begin{gather}\label{nov--1}
m_t=K_{-1}=\mathcal{J} \frac{\delta \mathcal{H}_{-1}}{\delta m}=\mathcal{J} \frac{\delta \mathcal{H}_C}{\delta m}=3\mathcal{J} m^{-\frac{1}{3}}.
\end{gather}

The SK equation
\begin{gather}\label{sk}
Q_\tau+Q_{yyyyy}+5Q Q_{yyy}+5Q_yQ_{yy}+5Q^2Q_y=0
\end{gather}
exhibits a generalized bi-Hamiltonian system, whose corresponding integrable hierarchy is ge\-nerated by a recursion operator $\bar{\mathcal{R}}=\bar{\mathcal{K}}\bar{\mathcal{J}}$, with
\begin{gather}\label{sk-op-k}
\bar{\mathcal{K}}=-\big( \partial_y^3+2 Q \partial_y+2 \partial_y Q\big)
\end{gather}
and
\begin{gather}\label{sk-op-j}
\bar{\mathcal{J}}=2\partial_y^3+2 \partial_y^2 Q \partial_y^{-1}+2 \partial_y^{-1} Q\partial_y^2+Q^2 \partial_y^{-1}+\partial_y^{-1} Q^2.
\end{gather}
As noted in \cite{fo}, $\bar{\mathcal{K}}$ maps the variational gradients of the conservation laws of the equation under consideration onto its symmetry groups, while~$\bar{\mathcal{J}}$ works in the opposite way.

\begin{Definition}
The equation $Q_\tau=\bar{K}[Q]$ is called a generalized bi-Hamiltonian system if there exist an implectic (Hamiltonian) operator $\bar{\mathcal{K}}$ and a functional $\bar{\mathcal{H}}_0$ such that
\begin{gather*}
\bar{K}[Q]=\bar{\mathcal{K}}\frac{\delta \bar{\mathcal{H}}_0}{\delta Q},
\end{gather*}
as well as a symplectic operator $\bar{\mathcal{J}}$ and a corresponding functional $\bar{\mathcal{H}}_1$ satisfying
\begin{gather*}
\bar{\mathcal{J}}\bar{K}[Q]=\frac{\delta \bar{\mathcal{H}}_1}{\delta Q}.
\end{gather*}
\end{Definition}

The term ``generalized bi-Hamiltonian system'' is taken from \cite{fo}, and refers to the fact that we do not assume any nondegeneracy or invertibility conditions for the operators $\bar{\mathcal{K}}$ and $\bar{\mathcal{J}}$. These are particular instances of the general notion of \emph{compatible pairs of Dirac structures}, whose properties are developed in Dorfman~\cite{Dorfman}.

Therefore, def\/ining
\begin{gather}\label{sk-H0}
\bar{\mathcal{H}}_0= \frac{1}{6}\int \big( Q^3-3Q_y^2\big) \mathrm{d}y\qquad \hbox{with }\qquad \frac{\delta \bar{\mathcal{H}}_0}{\delta Q}=\frac{1}{2}Q^2+Q_{yy},
\end{gather}
one f\/inds that the SK equation \eqref{sk} can be written as
\begin{gather*}
Q_\tau=\bar{K}_{1}=\bar{\mathcal{K}} \frac{\delta \bar{\mathcal{H}}_0}{\delta Q},
\end{gather*}
and the positive f\/lows of the SK hierarchy are generated by applying successively the recursion operator $\bar{\mathcal{R}}=\bar{\mathcal{K}} \bar{\mathcal{J}}$ to $\bar{K}_{1}$, namely
\begin{gather}\label{pohi-sk}
Q_{\tau}=\bar{K}_{n}=\big( \bar{\mathcal{K}} \bar{\mathcal{J}}\big)^{n-1}\bar{K}_{1},\qquad n=1, 2, \ldots.
\end{gather}
On the other hand, in the negative direction, note that the trivial function $f=0$ satisf\/ies the equation
\begin{gather*}
\bar{\mathcal{J}}\cdot f=\frac{\delta \bar{\mathcal{H}}_0}{\delta Q}.
\end{gather*}
Then the $n$-th negative f\/low is proposed to take the form
\begin{gather}\label{sk--n}
 \bar{\mathcal{R}}^{n} Q_{\tau}=0,\qquad n=1, 2, \ldots.
\end{gather}
Furthermore, it has been discovered in \cite{cq} that the recursion operator $\bar{\mathcal{R}}$ satisf\/ies the following decomposition
\begin{gather}
\bar{\mathcal{R}}=\bar{\mathcal{K}} \bar{\mathcal{J}} =-\big( \partial_y^3+2 Q \partial_y+2 \partial_y Q\big) \big(2\partial_y^3+2 \partial_y^2 Q \partial_y^{-1}+2 \partial_y^{-1} Q\partial_y^2+Q^2 \partial_y^{-1}+\partial_y^{-1} Q^2\big)\nonumber\\
 \hphantom{\bar{\mathcal{R}}=\bar{\mathcal{K}} \bar{\mathcal{J}}}{}
 =-2\big( \partial_y^4+5Q \partial_y^2+4Q_y \partial_y+Q_{yy}+4Q^2+2Q_y \partial_y^{-1}Q\big) \big(\partial_y^2+Q+Q_y\partial_y^{-1}\big).\label{sk-ro-1}
\end{gather}
This factorization result demonstrates that the f\/irst negative f\/low \eqref{sk-1'} in the SK hierarchy derived in \cite{hw2} satisf\/ies $\bar{\mathcal{R}}Q_{\tau}=0$. We have thus conf\/irmed the formulation \eqref{sk--n} for the negative f\/lows.

We are now in the position to establish the theorem which shows how the transforma\-tions~\eqref{recixy-nov} af\/fect the underlying Liouville correspondence between the Novikov and SK hie\-rarchies. In this theorem and hereafter, we denote, for a positive integer~$n$, the $n$-th equation in the positive and negative directions of the Novikov hierarchy by (Novikov)$_n$ and (Novikov)$_{-n}$, respectively, while the $n$-th positive and negative f\/lows of the SK hierarchy are denoted by (SK)$_{n}$ and (SK)$_{-n}$, respectively.

\begin{Theorem}\label{t1}
Under the Liouville transformation \eqref{recixy-nov}, for each nonzero integer $n\in \mathbb{Z}$, the {\rm(}Novikov{\rm)}$_{n}$ equation is mapped into the {\rm(}SK{\rm)}$_{-n}$ equation, and conversely.
\end{Theorem}

The proof of this theorem relies on the following two lemmas.

\begin{Lemma}\label{l3.1}
Let $m(t, x)$ and $Q(\tau, y)$ be related by the transformation \eqref{recixy-nov}. Then the following operator identities hold:
\begin{gather}
m^{-1} \big(1-\partial_x^2\big) m^{-\frac{1}{3}}=-\big(Q+\partial_y^2\big), \label{op1}\\
m^{-1} \mathcal{J} m^{-\frac{1}{3}}=\frac{1}{2}\partial_y \bar{\mathcal{J}} \partial_y, \label{op2}\\
m^{-\frac{4}{3}} \big(4\partial_x-\partial_x^3\big) m^{-\frac{2}{3}}=\bar{\mathcal{K}} .\label{op3}
\end{gather}
\end{Lemma}

\begin{proof}
{\textbf{(i)}}. In view of the transformation \eqref{recixy-nov}, one has $\partial_x=m^{\frac{2}{3}}\partial_y$. It follows that
\begin{gather*}
\partial_x^2 m^{-\frac{1}{3}}=m^{\frac{2}{3}} \partial_y m^{\frac{2}{3}} \partial_y=m \partial_y^2+\big( m^{-\frac{1}{3}}\big)_{xx},
\end{gather*}
where, by \eqref{recixy-nov}, a direct computation yields
\begin{gather*}
\big( m^{-\frac{1}{3}}\big)_{xx}=-\frac{1}{3} \big( m^{-\frac{4}{3}}m_x\big)_x
=\frac{1}{3} \left(\frac{4}{3} m^{-\frac{7}{3}} m_x^2-m^{-\frac{4}{3}} \left( \frac{4}{3} m^{-1} m_x^2-3 m-3 m^{\frac{7}{3}} Q\right) \right)\\
\hphantom{\big( m^{-\frac{1}{3}}\big)_{xx}}{} =m \big(m^{-\frac{4}{3}}+Q\big).
\end{gather*}
We thus arrive at
\begin{gather*}
m^{-1} \big(1-\partial_x^2\big) m^{-\frac{1}{3}}=m^{-1} \big( m^{-\frac{1}{3}}-\big(m \partial_y^2+m^{-\frac{1}{3}}+m Q\big)\big)=-\big(Q+\partial_y^2\big).
\end{gather*}

{\textbf{(ii)}}. Thanks to \eqref{op1}, we deduce that
\begin{gather*}
m^{-1} \big(1-\partial_x^2\big) m^{-1} \partial_x m^{-1} \big(1-\partial_x^2\big) m^{-\frac{1}{3}} \partial_y^{-1}=\big(Q+\partial_y^2\big) \partial_y \big(Q+\partial_y^2\big) \partial_y^{-1}=\frac{1}{2}\partial_y \bar{\mathcal{J}},
\end{gather*}
verifying \eqref{op2}.

{\textbf{(iii)}}. Using the transformations \eqref{recixy-nov} again, we f\/ind
 \begin{gather*}
\partial_x^2 m^{-\frac{2}{3}} =m^{\frac{2}{3}} \partial_y\left( \partial_y-\frac{2}{3} m^{-\frac{5}{3}} m_x \right)=m^{\frac{2}{3}} \left( \partial_y^2-\frac{2}{3} m^{-\frac{5}{3}} m_x \partial_y-\frac{2}{3} m^{-\frac{2}{3}} \big(m^{-\frac{5}{3}} m_x\big)_x\right) \\
\hphantom{\partial_x^2 m^{-\frac{2}{3}}}{} = \frac{2}{9} m^{-\frac{8}{3}} m_x^2 +2 m^{-\frac{2}{3}}+2 m^{\frac{2}{3}} Q-\frac{2}{3}m^{-1} m_x \partial_y+m^{\frac{2}{3}} \partial_y^2.
\end{gather*}
Hence,
\begin{gather*}
 m^{-\frac{4}{3}} \big(4\partial_x-\partial_x^3\big) m^{-\frac{2}{3}}= m^{-\frac{2}{3}} \partial_y (4-\partial_x^2) m^{-\frac{2}{3}}\\
\hphantom{m^{-\frac{4}{3}} \big(4\partial_x-\partial_x^3\big) m^{-\frac{2}{3}}}{}
 = m^{-\frac{2}{3}} \partial_y \left( 2 m^{-\frac{2}{3}} -\frac{2}{9} m^{-\frac{8}{3}} m_x^2 -2 m^{\frac{2}{3}} Q+\frac{2}{3} m^{-1} m_x \partial_y-m^{\frac{2}{3}} \partial_y^2\right)\\
 \hphantom{m^{-\frac{4}{3}} \big(4\partial_x-\partial_x^3\big) m^{-\frac{2}{3}}}{}
 = 2 m^{-\frac{4}{3}} \big(m^{-\frac{2}{3}}\big)_x-\frac{2}{9} m^{-\frac{4}{3}} \big(m^{-\frac{8}{3}} m_x^2\big)_x-2 m^{-\frac{4}{3}} \big(m^{\frac{2}{3}}\big)_x Q-2 Q_y\\
\hphantom{m^{-\frac{4}{3}} \big(4\partial_x-\partial_x^3\big) m^{-\frac{2}{3}}=}{}
 +\left( 2 m^{-\frac{4}{3}}-\frac{2}{9} m^{-\frac{10}{3}} m_x^2+\frac{2}{3} m^{-\frac{4}{3}} \big(m^{-1} m_x\big)_x-2 Q\right) \partial_y\\
\hphantom{m^{-\frac{4}{3}} \big(4\partial_x-\partial_x^3\big) m^{-\frac{2}{3}}=}{}
+\left( m^{-\frac{2}{3}} \big(m^{\frac{2}{3}}\big)_x-\frac{2}{3} m^{-1} m_x\right) \partial_y^2-\partial_y^3,
 \end{gather*}
then \eqref{op3} follows.
\end{proof}

The relationship between the recursion operator for the Novikov hierarchy and the adjoint operator of recursion operator for the SK hierarchy is given by the following result.

\begin{Lemma}\label{lem2.2}
Under the transformation \eqref{recixy-nov}, the relation
\begin{gather}\label{l2.2}
m^{-1}\big(\mathcal{J} \mathcal{K}^{-1}\big)^n m=\partial_y\big(\bar{\mathcal{J}}\bar{\mathcal{K}}\big)^n\partial_y^{-1}
\end{gather}
holds for each integer $n\geq 1$.
\end{Lemma}

\begin{proof}
 We prove \eqref{l2.2} by induction on $n$. First, using the inverse operator $\mathcal{K}^{-1}$ along with~\eqref{recixy-nov}, we deduce from \eqref{op2} and~\eqref{op3} that
 \begin{gather*}
m^{-1} \mathcal{J} \mathcal{K}^{-1} m=2 m^{-1} \mathcal{J} m^{-\frac{1}{3}} \partial_x^{-1} m^{-\frac{2}{3}} \big(4\partial_x-\partial_x^3\big) m^{-\frac{2}{3}} \partial_x^{-1} m^{\frac{2}{3}}=\partial_y \bar{\mathcal{J}} \bar{\mathcal{K}} \partial_y^{-1},
\end{gather*}
which shows that \eqref{l2.2} holds for $n=1$. Next, we assume that \eqref{l2.2} holds for $n=k$, say
\begin{gather*}\label{op-l2}
m^{-1}\big(\mathcal{J} \mathcal{K}^{-1}\big)^k m=\partial_y\big(\bar{\mathcal{J}}\bar{\mathcal{K}}\big)^k\partial_y^{-1}.
\end{gather*}
Then for $n=k+1$, thanks to the result when $n=1$, one has
\begin{gather*}
m^{-1}\big(\mathcal{J} \mathcal{K}^{-1}\big)^{k+1} m =m^{-1} \mathcal{J} \mathcal{K}^{-1} \big(\mathcal{J} \mathcal{K}^{-1}\big)^k m=\partial_y\big(\bar{\mathcal{J}}\bar{\mathcal{K}}\big)^{k+1}\partial_y^{-1}.
\end{gather*}
This completes the induction step, and thus proves the lemma.
\end{proof}

\begin{proof} [\textbf{Proof of Theorem \ref{t1}}]
The proof of Theorem \ref{t1} contains two parts:

{\textbf{(i)}}. Let us begin with the (Novikov)$_{-n}$ equation for $n\geq 1$. First, since the relation \eqref{recixy-nov} can be rewritten as
\begin{gather}\label{reciqm1}
Q=-m^{-1} \big(1-\partial_x^2\big) m^{-\frac{1}{3}},
\end{gather}
we deduce from \eqref{op1} that the f\/irst negative f\/low \eqref{nov--1} of the Novikov hierarchy satisf\/ies
\begin{gather*}%\label{nov--1Q}
m_t=K_{-1}=-3 \big(1-\partial_x^2\big) m^{-1} Q_x=3m\big(Q+\partial_y^2\big) Q_y=3m \partial_y \left(\frac{1}{2}Q^2+Q_{yy}\right).
\end{gather*}
Hence, the $n$-th member in the negative hierarchy takes the form
\begin{gather}\label{nov--nQ}
m_t=K_{-n}=\big(\mathcal{J} \mathcal{K}^{-1}\big)^{n-1} K_{-1}
= 3\big(\mathcal{J} \mathcal{K}^{-1}\big)^{n-1} m \partial_y \left(\frac{1}{2}Q^2+Q_{yy}\right),\qquad n=1, 2,\ldots.\!\!\!
\end{gather}

Next, suppose that $m(t, x)$ is the solution of the equation~\eqref{nov--nQ}. We calculate the $t$-derivative of the corresponding function $Q(\tau, y)$ def\/ined in~\eqref{reciqm1}. More precisely, we deduce that, on the one hand,
\begin{gather*}
Q_t=Q_\tau+Q_y\int^x \big( m^{\frac{2}{3}}(t, \xi)\big)_t \mathrm{d}\xi
=Q_\tau+\frac{2}{3}Q_y\partial_x^{-1}m^{-\frac{1}{3}}m_t=Q_\tau+\frac{2}{3}Q_y\partial_y^{-1}m^{-1}m_t,
\end{gather*}
and on the other hand, in view of \eqref{reciqm1},
\begin{gather*}
Q_t=m^{-2} m_t \big(1-\partial_x^2\big) m^{-\frac{1}{3}}+\frac{1}{3}m^{-1} \big(1-\partial_x^2\big) m^{-\frac{4}{3}}m_t=-m^{-1}\left( Q-\frac{1}{3}\big(1-\partial_x^2\big) m^{-\frac{4}{3}}\right) m_t.
\end{gather*}
Hence, combining the preceding two equations, we arrive at
\begin{gather}\label{t3.1-1}
Q_\tau=\left( -\frac{2}{3} Q_y\partial_y^{-1}m^{-1}-m^{-1}Q+\frac{1}{3}m^{-1}\big(1-\partial_x^2\big) m^{-\frac{4}{3}}\right) m_t=\frac{1}{3} \bar{\mathcal{K}}\partial_y^{-1}m^{-1}m_t,
\end{gather}
where we have made use of the formula \eqref{op1}.

Finally, according to Lemma \ref{lem2.2}, we deduce that if $m(t, x)$ is the solution of the (Novikov)$_{-n}$ equation \eqref{nov--nQ}, the corresponding $Q(\tau, y)$ satisf\/ies
 \begin{gather*}
Q_\tau =\bar{\mathcal{K}}\partial_y^{-1}m^{-1}\big(\mathcal{J} \mathcal{K}^{-1}\big)^{n-1} m \partial_y\left(\frac{1}{2}Q^2+Q_{yy}\right)\\
\hphantom{Q_\tau}{} =\bar{\mathcal{K}}\big(\bar{\mathcal{J}} \bar{\mathcal{K}}\big)^{n-1} \left(\frac{1}{2}Q^2+Q_{yy}\right)=\bar{\mathcal{R}}^{n-1}\bar{K}_1=\bar{K}_n.
\end{gather*}
This immediately implies that $Q(\tau, y)$ solves the (SK)$_{n}$ equation~\eqref{pohi-sk}.

Conversely, if $Q(\tau, y)$ is a solution of the (SK)$_{n}$ equation for $n\geq 1$, since the transforma\-tion~\eqref{recixy-nov} is a bijection, tracing the previous steps backwards suf\/f\/ices to verify that the reverse argument is also true.

{\textbf{(ii)}}. Now, we focus our attention on the (Novikov)$_{n}$ equation for $n\geq 1$, which can be written as
\begin{gather}\label{nov-n}
m_t=K_n=\big( \mathcal{K}\mathcal{J}^{-1}\big)^{n-1} \mathcal{K} \frac{\delta \mathcal{H}_0}{\delta m}=9\big( \mathcal{K}\mathcal{J}^{-1}\big)^{n-1} m^{\frac{1}{3}} \partial_x m^{\frac{2}{3}} \big(4\partial_x-\partial_x^3\big)^{-1} m^{\frac{2}{3}} \partial_x m^{\frac{1}{3}} u.
\end{gather}
 Plugging it into \eqref{t3.1-1}, we f\/ind
\begin{gather*}
Q_\tau=3\bar{\mathcal{K}}\partial_y^{-1}m^{-1}\big( \mathcal{K}\mathcal{J}^{-1}\big)^{n-1} m^{\frac{1}{3}} \partial_x m^{\frac{2}{3}} \big(4\partial_x-\partial_x^3\big)^{-1} m^{\frac{2}{3}} \partial_x m^{\frac{1}{3}} u.
\end{gather*}
As a consequence, the operator factorization identity \eqref{sk-ro-1}, when combined with~\eqref{l2.2}, allows us to deduce that, for each $n\geq 1$, if $m(t, x)$ solves the (Novikov)$_n$ equation \eqref{nov-n}, then for the operator $\mathcal{B}$ def\/ined by
\begin{gather*}
\mathcal{B}=-2\big( \partial_y^4+5Q \partial_y^2+4Q_y \partial_y+Q_{yy}+4Q^2+2Q_y \partial_y^{-1}Q\big),
\end{gather*}
the corresponding $Q(\tau, y)$ satisf\/ies
\begin{gather*}
 \bar{\mathcal{R}}^n Q_\tau=\mathcal{B} \big(\partial_y^2+Q+Q_y\partial_y^{-1}\big) \big(\bar{\mathcal{K}}\bar{\mathcal{J}}\big)^{n-1} Q_\tau\\
\hphantom{\bar{\mathcal{R}}^n Q_\tau}{} =3\mathcal{B} \partial_y\big( \partial_y^2\!+Q\big) \partial_y^{{-}1}\big(\bar{\mathcal{K}}\bar{\mathcal{J}}\big)^{n{-}1}\bar{\mathcal{K}}\partial_y^{{-}1}m^{{-}1} \big( \mathcal{K}\mathcal{J}^{{-}1}\big)^{n{-}1}
%\\ \hphantom{\bar{\mathcal{R}}^n Q_\tau=}{}\cdot
 m^{\frac{1}{3}} \partial_x m^{\frac{2}{3}}\big(4\partial_x\!- \partial_x^3\big)^{{-}1} m^{\frac{2}{3}} \partial_x m^{\frac{1}{3}} u\\
\hphantom{\bar{\mathcal{R}}^n Q_\tau}{}
=3\mathcal{B} \partial_y \big( \partial_y^2+Q\big) \partial_y^{-1} \bar{\mathcal{K}} m^{\frac{2}{3}} \big(4\partial_x-\partial_x^3\big)^{-1} m^{\frac{2}{3}} \partial_x m^{\frac{1}{3}} u\\
\hphantom{\bar{\mathcal{R}}^n Q_\tau}{} =3\mathcal{B} \partial_y \big( \partial_y^2+Q\big) m^{\frac{1}{3}} u=-3\mathcal{B} \partial_y\cdot 1=0,
\end{gather*}
where we have made use of the operator identity \eqref{op1}. This immediately reveals that $Q(\tau, y)$ solves the (SK)$_{-n}$ equation~\eqref{sk--n}. We thus prove that, for each $n\geq 1$ the (Novikov)$_n$ equation is mapped into the (SK)$_{-n}$ equation under the transformation~\eqref{recixy-nov}.

In analogy with the proof of part (i), the reverse argument follows from the fact that \eqref{recixy-nov} is a bijection.
\end{proof}

\subsection[The correspondence between the Hamiltonian conservation laws of the Novikov and SK equations]{The correspondence between the Hamiltonian conservation laws\\ of the Novikov and SK equations}\label{section2.3}

According to Magri's theorem, one can also recursively construct
an inf\/inite hierarchy of Hamiltonian conservation laws of any bi-Hamiltonian structure. In particular, for the Novikov equation~\eqref{nov},
at the $n$-th stage we determine the Hamiltonian conservation laws $\mathcal{H}_{n}$ satisfying the recursive formula
\begin{gather}\label{novclhie}
\mathcal{K} \frac{\delta \mathcal{H}_{n-1}}{\delta m}=\mathcal{J} \frac{\delta \mathcal{H}_n}{\delta m},\qquad n\in \mathbb{Z},
\end{gather}
where $\mathcal{K}$ and $\mathcal{J}$ are the two compatible Hamiltonian operators~\eqref{nov-kj} admitted by the Novikov equation.
On the other hand, the recursive formula
\begin{gather}\label{skclhie}
\bar{\mathcal{J}} \bar{\mathcal{K}} \frac{\delta \bar{\mathcal{H}}_{n-1}}{\delta Q}=\frac{\delta \bar{\mathcal{H}}_n}{\delta Q},\qquad n\in \mathbb{Z},
\end{gather}
formally provides an inf\/inite collection of Hamiltonian conservation laws for the SK equation~\eqref{sk},
using the operator pair $\bar{\mathcal{K}}$ and $\bar{\mathcal{J}}$ given in~\eqref{sk-op-k} and \eqref{sk-op-j}.

In this subsection we investigate the relationship between the two hierarchies of Hamiltonian conservation laws
$\{\mathcal{H}_n\}$ and $\{\bar{\mathcal{H}}_n\}$. Let us begin with two preliminary lemmas.

\begin{Lemma}
Let $\{\mathcal{H}_n\}$ and $\{\bar{\mathcal{H}}_n\}$ be the hierarchies of Hamiltonian conservation laws of the Novikov and SK equations, respectively. Then, for each $n\in \mathbb{Z}$, their corresponding variational derivatives satisfy the relation
\begin{gather}\label{l3.3-1}
\frac{\delta \bar{ \mathcal{H}}_{n}}{\delta Q}=\frac{1}{3}\partial_x^{-1}m^{-\frac{1}{3}}\mathcal{K}\frac{\delta \mathcal{H}_{-(n+2)}}{\delta m}.
\end{gather}
\end{Lemma}

\begin{proof}
The proof relies on an induction argument. First of all, since
\begin{gather*}
\frac{\delta \mathcal{H}_{-2}}{\delta m} =\mathcal{K}^{-1}\mathcal{J}\frac{\delta \mathcal{H}_{-1}}{\delta m}=3\mathcal{K}^{-1}\big(1-\partial_x^2\big)m^{-1}\partial_xm^{-1}\big(1-\partial_x^2\big) m^{-\frac{1}{3}}\\
\hphantom{\frac{\delta \mathcal{H}_{-2}}{\delta m}}{} =-3\mathcal{K}^{-1}\big(1-\partial_x^2\big)m^{-1} Q_x=3\mathcal{K}^{-1}m\big(Q+\partial_y^2\big) Q_y =3\mathcal{K}^{-1}m\partial_y\frac{\delta \bar{\mathcal{H}}_{0}}{\delta Q},
\end{gather*}
by \eqref{recixy-nov}, \eqref{op1}, and \eqref{reciqm1}, and so, clearly, \eqref{l3.3-1} holds for $n=0$.

Suppose by induction, that \eqref{l3.3-1} holds for $n=k$ with $k\geq 0$,
say
\begin{gather*}\label{vd-sk}
\frac{\delta \bar{ \mathcal{H}}_{k}}{\delta Q}=\frac{1}{3}\partial_x^{-1}m^{-\frac{1}{3}}\mathcal{K}\frac{\delta \mathcal{H}_{-(k+2)}}{\delta m}.
\end{gather*}
Then, for $n=k+1$, by \eqref{novclhie} and \eqref{skclhie}, we deduce that
\begin{gather*}
\begin{split}&
\frac{\delta \bar{ \mathcal{H}}_{k+1}}{\delta Q} =\bar{\mathcal{J}}\bar{\mathcal{K}}\frac{\delta \bar{ \mathcal{H}}_{k}}{\delta Q}=\frac{1}{3}\bar{\mathcal{J}}\bar{\mathcal{K}}\partial_x^{-1}m^{-\frac{1}{3}}\mathcal{K}\frac{\delta \mathcal{H}_{-(k+2)}}{\delta m}\\
& \hphantom{\frac{\delta \bar{ \mathcal{H}}_{k+1}}{\delta Q}}{}
=\frac{1}{3}\bar{\mathcal{J}}\bar{\mathcal{K}}\partial_y^{-1}m^{-1}\mathcal{K}\mathcal{J}^{-1}\mathcal{K}\frac{\delta \mathcal{H}_{-(k+3)}}{\delta m}=\frac{1}{3}\partial_x^{-1}m^{-\frac{1}{3}}\mathcal{K}\frac{\delta \mathcal{H}_{-(k+3)}}{\delta m},
\end{split}
\end{gather*}
where we have made use of Lemma \ref{lem2.2} with $n=1$. This verif\/ies \eqref{l3.3-1} for $n\geq 0$, completing the f\/irst step.

Next, in the case of $n=-1$, it follows from $\delta \mathcal{H}_{-1}/\delta m=3m^{-1/3}$ that
\begin{gather*}
\frac{1}{3}\bar{\mathcal{K}}\partial_x^{-1}m^{-\frac{1}{3}}\mathcal{K}\frac{\delta \mathcal{H}_{-1}}{\delta m}=\frac{1}{2}m^{-\frac{2}{3}}\partial_x\cdot 1=0,
\end{gather*}
which, together with the fact that $\bar{\mathcal{K}}\cdot (\delta\bar{\mathcal{H}}_{-1}/\delta Q)=0$ shows that \eqref{l3.3-1} holds for $n=-1$.

Finally, to prove \eqref{l3.3-1} holds for all $n\leq -1$, we assume that \eqref{l3.3-1} holds for $n=k$. Then for $n=k-1$, using the recursive formulae \eqref{novclhie} and \eqref{skclhie} and Lemma~\ref{lem2.2} with $n=1$ again, we infer that
\begin{gather*}
\frac{\delta \mathcal{H}_{-(k-1)}}{\delta m}= \mathcal{J}^{-1}\mathcal{K}\frac{\delta \mathcal{H}_{-k}}{\delta m}=3\mathcal{J}^{-1}m\partial_y\frac{\delta \bar{ \mathcal{H}}_{k-2}}{\delta Q} =3\mathcal{J}^{-1}m\partial_y\bar{\mathcal{J}}\bar{\mathcal{K}}\frac{\delta \bar{ \mathcal{H}}_{k-3}}{\delta Q}=3\mathcal{K}^{-1}m^{\frac{1}{3}}\partial_x\frac{\delta \bar{ \mathcal{H}}_{k-3}}{\delta Q},\!
\end{gather*}
which establishes the induction step for $n\leq -1$ and thus proves the lemma in general.
\end{proof}

In order to establish the correspondence between Hamiltonian conservation laws admitted by the Novikov and SK equations, we require the formula for the change of variational derivatives.

\begin{Lemma}\label{lem2.4}
Let $m(t, x)$ and $Q(\tau, y)$ be related by the transformations \eqref{recixy-nov}. If $\mathcal{H}(m)=\bar{\mathcal{H}}(Q)$, then
\begin{gather}\label{vdqm}
\frac{\delta \mathcal{H}}{\delta m}=\frac{1}{3} m^{-\frac{1}{3}}\partial_y^{-1}\bar{\mathcal{K}}\frac{\delta \bar{\mathcal{H}}}{\delta Q},
\end{gather}
where $\bar{\mathcal{K}}$ is the Hamiltonian operator \eqref{sk-op-k} admitted by the SK equation \eqref{sk}.
\end{Lemma}

\begin{proof}
First of all, motivated by \eqref{recixy-nov} and \eqref{reciqm1}, we introduce
\begin{gather*}
F[m(t, x)]\equiv -m^{-1}\big(1-\partial_x^2\big)m^{-\frac{1}{3}}=Q(\tau, y).
\end{gather*}
Then the Fr\'echet derivative of $F[m]$ is
\begin{gather*}
\mathcal{D}_{F[m]} =\frac{4}{3}m^{-\frac{7}{3}}-m^{-2}\big(m^{-\frac{1}{3}}\big)_{xx}-\frac{1}{3}m^{-1}\partial_x^2 m^{-\frac{4}{3}}\\
\hphantom{\mathcal{D}_{F[m]}}{} =-m^{-1}Q-\frac{1}{3}\big(Q+\partial_y^2\big)m^{-1}=-\frac{1}{3}\big(4Q+\partial_y^2\big)m^{-1}.
\end{gather*}
On the other hand,
\begin{gather*}\label{l4-1}
\frac{\mathrm{d}}{\mathrm{d}\epsilon}\Big|_{\epsilon=0}F[m+\epsilon \rho]=Q_y\frac{\mathrm{d}}{\mathrm{d}\epsilon}\Big|_{\epsilon=0}y(m+\epsilon \rho)+\frac{\mathrm{d}}{\mathrm{d}\epsilon}\Big|_{\epsilon=0}^{y \, {\rm f\/ixed}}F[m+\epsilon \rho],
\end{gather*}
where, by \eqref{recixy-nov},
\begin{gather*}\frac{\mathrm{d}}{\mathrm{d}\epsilon}\Big|_{\epsilon=0}y(m+\epsilon \rho)=\frac{2}{3}\partial_x^{-1} m^{-\frac{1}{3}}\rho.
\end{gather*}

Next, it follows from
\begin{gather*}
\frac{\mathrm{d}}{\mathrm{d}\epsilon}\Big|_{\epsilon=0}F[m+\epsilon \rho]=\mathcal{D}_{F[m]}(\rho)=-\frac{1}{3}\big(4Q+\partial_y^2\big)m^{-1}\rho
\end{gather*} that
\begin{gather*}
\frac{\mathrm{d}}{\mathrm{d}\epsilon}\Big|_{\epsilon=0}^{y \, {\rm f\/ixed}}F[m+\epsilon \rho]=-\frac{1}{3}\big(4Q+\partial_y^2\big)m^{-1}\rho-\frac{2}{3}Q_y\partial_y^{-1} m^{-1}\rho=\frac{1}{3}\bar{\mathcal{K}}\partial_y^{-1} m^{-1}\rho.
\end{gather*}

Finally, the assumption of the lemma implies that
\begin{gather*}
\frac{\mathrm{d}}{\mathrm{d}\epsilon}\Big|_{\epsilon=0} \mathcal{H}(m+\epsilon \rho)=\frac{\mathrm{d}}{\mathrm{d}\epsilon}\Big|_{\epsilon=0} \bar{\mathcal{H}}\left(F[m+\epsilon \rho]\right).
\end{gather*}
According to the usual def\/inition of the variational derivative, we have, on the one hand,
\begin{gather}\label{l4-2}
\frac{\mathrm{d}}{\mathrm{d}\epsilon}\Big|_{\epsilon=0}\mathcal{H}(m+\epsilon \rho)=\int \frac{\delta \mathcal{H}}{\delta m} \cdot \rho \mathrm{d}x.
\end{gather}
On the other hand, using the fact that $\bar{\mathcal{K}}$ is skew-symmetric, we infer that
\begin{gather*}
\frac{\mathrm{d}}{\mathrm{d}\epsilon}\Big|_{\epsilon=0} \bar{\mathcal{H}} (F[m+\epsilon \rho] )
=\int \frac{\delta \bar{\mathcal{H}}}{\delta Q} \cdot\frac{\mathrm{d}}{\mathrm{d}\epsilon}\Big|_{\epsilon=0}^{y \, {\rm f\/ixed}}F[m+\epsilon \rho] \mathrm{d}y=\frac{1}{3} \int \frac{\delta \bar{\mathcal{H}}}{\delta Q} \cdot\bar{\mathcal{K}}\partial_y^{-1} m^{-1}\rho \mathrm{d}y\\
\hphantom{\frac{\mathrm{d}}{\mathrm{d}\epsilon}\Big|_{\epsilon=0} \bar{\mathcal{H}} (F[m+\epsilon \rho] ) }{}
=\frac{1}{3} \int m^{\frac{2}{3}} \big(\bar{\mathcal{K}}\partial_y^{-1} m^{-1}\big)^*\frac{\delta \bar{\mathcal{H}}}{\delta Q} \cdot \rho \mathrm{d}x=\frac{1}{3} \int m^{-\frac{1}{3}} \partial_y^{-1}\bar{\mathcal{K}}\frac{\delta \bar{\mathcal{H}}}{\delta Q} \cdot \rho \mathrm{d}x,
\end{gather*}
which, in comparison with \eqref{l4-2} verif\/ies \eqref{vdqm}, proving the lemma.
\end{proof}

Finally, referring back to the form of the Hamiltonian operator $\mathcal{K}$, one has
\begin{gather*}
\mathcal{K}^{-1}m^{\frac{1}{3}}\partial_x=2m^{-\frac{1}{3}}\partial_x^{-1}m^{-\frac{2}{3}}\big(4\partial_x-\partial_x^3\big)
m^{-\frac{2}{3}}=2m^{-\frac{1}{3}}\partial_y^{-1}\bar{\mathcal{K}}.
\end{gather*}
It follows that the relation \eqref{l3.3-1} can be written in an equivalent form, namely
\begin{gather}\label{l3.3-1'}
\frac{\delta \mathcal{H}_{n}}{\delta m}=3\mathcal{K}^{-1}m^{\frac{1}{3}}\partial_x\frac{\delta \bar{ \mathcal{H}}_{-(n+2)}}{\delta Q}=6m^{-\frac{1}{3}}\partial_y^{-1}\bar{\mathcal{K}}\frac{\delta \bar{ \mathcal{H}}_{-(n+2)}}{\delta Q}.
\end{gather}
Therefore, subject to the hypothesis of Lemma \ref{lem2.4}, if we def\/ine the functional
\begin{gather*}
\mathcal{G}_{l}(Q)\equiv \mathcal{H}_{n}(m),
\end{gather*}
for some $l \in \mathbb{Z}$, then Lemma \ref{lem2.4} allows us to conclude that, for each $n\in \mathbb{Z}$,
\begin{gather*}
\frac{\delta \mathcal{H}_{n}}{\delta m}= \frac{1}{3}m^{-\frac{1}{3}} \partial_y^{-1}\bar{\mathcal{K}}\frac{\delta \mathcal{G}_{l}}{\delta Q}.
\end{gather*}
This, when combined with \eqref{l3.3-1'}, immediately leads to
\begin{gather*} \mathcal{G}_{l}(Q)=18 \bar{\mathcal{H}}_{-(n+2)}(Q),\end{gather*}
and then
\begin{gather*}
\mathcal{H}_{n}(m)=18 \bar{\mathcal{H}}_{-(n+2)}(Q)
\end{gather*}
follows. We thus conclude that there exsits an one-to-one correspondence between the sequences of the Hamiltonian conservation laws admitted by the Novikov and SK equations.

Indeed, we have proved the following theorem.

\begin{Theorem}\label{t2}
Under the Liouville transformation \eqref{recixy-nov}, for each $n\in \mathbb{Z}$, the Hamiltonian conservation law $\bar{\mathcal{H}}_n(Q)$ of the SK equation is related to the Hamiltonian conservation law $\mathcal{H}_{-n}(m)$ of the Novikov equation, according to the following identity
\begin{gather*}
\mathcal{H}_{n}(m)=18 \bar{\mathcal{H}}_{-(n+2)}(Q), \qquad n\in \mathbb{Z}.
\end{gather*}
\end{Theorem}

For instance, in the case of $n=2$,
\begin{gather*}\label{cl-rela}
\frac{\delta \mathcal{H}_{-2}}{\delta m}=\mathcal{K}^{-1} \mathcal{J}\frac{\delta \mathcal{H}_{-1}}{\delta m},
\end{gather*}
which can be expressed in terms of $Q$ according to \eqref{reciqm1}, say
\begin{gather*}
\frac{\delta \mathcal{H}_{-2}}{\delta m}
=6m^{-\frac{1}{3}}\partial_y^{-1}\bar{\mathcal{K}}\left(\frac{1}{2}Q^2+Q_{yy}\right)=-6m^{-\frac{1}{3}}\left( Q_{yyyy}+5QQ_{yy}+\frac{5}{3}Q^3\right).
\end{gather*}
As a consequence,
\begin{gather*}
\mathcal{H}_{-2}(m)=9\int m^{\frac{2}{3}} \left( \frac{1}{5}Q_{yyyy}+QQ_{yy}+\frac{1}{3}Q^3 \right) \mathrm{d}x=3\int \big( Q^3-3Q_y^2\big) \mathrm{d}y,
\end{gather*}
with $Q$ being determined by \eqref{recixy-nov} and \eqref{reciqm1}, which, when compared with \eqref{sk-H0}, shows that $\mathcal{H}_{-2}(m)=18\bar{\mathcal{H}}_{0}(Q)$, in accordance with Theorem~\ref{t2}.

\section{The correspondence between the DP and KK hierarchies}\label{section3}

\subsection[A Liouville transformation between the isospectral problems of the DP and KK hierarchies]{A Liouville transformation between the isospectral problems\\ of the DP and KK hierarchies}\label{section3.1}

The Lax pair for the DP equation
\begin{gather}\label{dp}
n_t=vn_x+3v_xn,\qquad n=v-v_{xx},
\end{gather}
takes the form \cite{hw2}
\begin{gather}\label{iso-dp}
\mathbf{\Psi}_x =\begin{pmatrix} 0 &1 & 0\\
 0 &0 & 1\\
 -\lambda n &1 & 0
 \end{pmatrix} \mathbf{\Psi},\qquad
\mathbf{\Psi}=\begin{pmatrix}\psi_1 \\ \psi_2 \\ \psi_3\end{pmatrix},
\end{gather}
and \begin{gather*}%\label{iso-dpt}
\mathbf{\Psi}_t
=\begin{pmatrix} v_x &-v & -\lambda^{-1}\\
 v &-\lambda^{-1} & -v\\
 \lambda vn+v_x & 0 & -\lambda^{-1}-v_x
 \end{pmatrix}\mathbf{\Psi}.
\end{gather*}
These can be rewritten in scalar form by setting $\Psi=\psi_1$, namely
\begin{gather*}%\label{iso-dp-scalar}
\Psi_{xxx}-\Psi_{x} +\lambda n \Psi=0, \qquad \Psi_{t}+\lambda^{-1}\Psi_{xx} +v\Psi_{x}-v_x \Psi=0.
\end{gather*}

Consider the KK equation
\begin{gather}\label{kk}
P_\tau+P_{yyyyy}+20PP_{yyy}+50P_yP_{yy}+80P^2P_y=0.
\end{gather}
It has been shown in \cite{hw2} that the Lax pair for the f\/irst negative f\/low of its associated hierarchy is
 \begin{gather}\label{iso-kk}
\Phi_{yyy}+4P\Phi_{y}+2P_y\Phi=\mu \Phi
\end{gather}and
\begin{gather}\label{laxt-kk}\Phi_{\tau}+\mu^{-1}\left( U\Phi_{yy}-\frac{1}{2}U_y\Phi_{y} +\frac{1}{6}(U_{yy}+16PU)\Phi\right)=0,
\end{gather}
which is a reduction of a $(2+1)$-dimensional non-isospectral Lax pair given in \cite{gp}. The compatibility condition for (\ref{iso-kk}) and (\ref{laxt-kk}) gives rise to
\begin{gather}\label{kk--1'}
P_\tau=\frac{3}{4}U_{y},\qquad \mathcal{A} U=0,
\end{gather}
where $\mathcal{A}$ is the f\/ifth-order operator
\begin{gather*}
\mathcal{A}=\partial_y^5+6\big(\partial_y P \partial_y^2+\partial_y^2 P \partial_y\big)+4\big(\partial_y^3 P+P \partial_y^3\big)+32\big(\partial_y P^2+P^2 \partial_y\big).
\end{gather*}

In analogy with the Liouville correspondence between the Novikov equation and the f\/irst negative f\/low of the SK hierarchy, there exists a similar correspondence between the DP equation and the f\/irst negative f\/low of the KK hierarchy. In fact, it has been found \cite{dhh} that the following coordinate transformations
\begin{gather*}
\mathrm{d} y=n^{\frac{1}{3}} \mathrm{d}x+n^{\frac{1}{3}} v^2 \mathrm{d}t, \qquad \mathrm{d} \tau=\mathrm{d} t,
\end{gather*}
together with
\begin{gather*}
\Psi=n^{-\frac{1}{3}}\Phi,\qquad \lambda=-\mu, \qquad P=\frac{1}{4}\left(\frac{7}{9}n^{-\frac{8}{3}}n_x^2-\frac{2}{3}n^{-\frac{5}{3}}n_{xx}-n^{-\frac{2}{3}}\right)
\end{gather*}
will convert the scalar form of the isospectral problem (\ref{iso-dp}) into (\ref{iso-kk}).

As before, in this section we investigate the Liouville correspondence between the DP and KK hierarchies. More precisely, the respective f\/lows in the two hierarchies are related by the Liouville transformations
\begin{gather}\label{recixy-dp}
y=\int^x n^{\frac{1}{3}}(t, \xi) \mathrm{d}\xi, \qquad \tau=t,
\end{gather}
and
\begin{gather}\label{recipn}
P=\frac{1}{4}\left(\frac{7}{9}n^{-\frac{8}{3}}n_x^2-\frac{2}{3}n^{-\frac{5}{3}}n_{xx}-n^{-\frac{2}{3}}\right)
=\frac{1}{4}n^{-\frac{1}{2}}\big(4\partial_x^2-1\big) n^{-\frac{1}{6}}.
\end{gather} In addition, the relationship between the Hamiltonian conservation laws for the DP hierarchy and those for the KK hierarchy is also clarif\/ied.

\subsection{The correspondence between the DP and KK hierarchies}\label{section3.2}

The DP equation \eqref{dp} is also a bi-Hamiltonian system \cite{dhh}
\begin{gather*}
n_t=G_1=\mathcal{L} \frac{\delta \mathcal{E}_0}{\delta n}=\mathcal{D} \frac{\delta \mathcal{E}_1}{\delta n},\qquad n=v-v_{xx},
\end{gather*}
where
\begin{gather}\label{dp-ld}
\mathcal{L}=n^{\frac{2}{3}} \partial_x n^{\frac{1}{3}} \big(\partial_x-\partial_x^3\big)^{-1} n^{\frac{1}{3}} \partial_x n^{\frac{2}{3}}, \qquad
\mathcal{D}=\partial_x \big(1-\partial_x^2\big) \big(4-\partial_x^2\big)
\end{gather}
are a pair of compatible Hamiltonian operators, and the corresponding Hamiltonian functionals are
\begin{gather*}\label{dp-e0}
\mathcal{E}_0=\frac{9}{2}\int n \mathrm{d}x,\qquad \mathcal{E}_1= \frac{1}{6}\int u^3 \mathrm{d}x.
\end{gather*}

Applying the recursion operator $\widetilde{\mathcal{R}}=\mathcal{L}\mathcal{D}^{-1}$ successively to the initial symmetry $n_t=G_1$ gives rise to an inf\/inite hierarchy
\begin{gather*}
n_t=G_l=\mathcal{L} \frac{\delta \mathcal{E}_{l-1}}{\delta n}=\mathcal{D} \frac{\delta \mathcal{E}_l}{\delta n},\qquad l\in \mathbb{Z},
\end{gather*}
of commuting bi-Hamiltonian f\/lows and consequent conservation laws $\mathcal{E}_{l}$. As far as the associated negative f\/lows are concerned, noting that \begin{gather*}
G_0=\mathcal{D} \frac{\delta \mathcal{E}_0}{\delta n}=0,
\end{gather*}
and $\mathcal{L}$ admits the Casimir functional
\begin{gather*}
\mathcal{E}_C=18 \int n^{\frac{1}{3}} \mathrm{d}x\qquad \hbox{with}\qquad \frac{\delta \mathcal{E}_C}{\delta n}=6n^{-\frac{2}{3}}.
\end{gather*}
Therefore, we conclude that the f\/irst negative f\/low of the DP hierarchy is the Casimir equation
\begin{gather*}
n_t=G_{-1}=\mathcal{D} \frac{\delta \mathcal{E}_{C}}{\delta n}=6\mathcal{D} n^{-\frac{2}{3}}
\end{gather*}
and applying $\widetilde{\mathcal{R}}^{-1}=\mathcal{D}\mathcal{L}^{-1}$ successively to it produces the hierarchy of negative f\/lows, in which the $l$-th member takes the form
\begin{gather*}
n_t=G_{-l}=6\big( \mathcal{D}\mathcal{L}^{-1}\big)^{l-1} \mathcal{D} n^{-\frac{2}{3}},\qquad l=1, 2,\ldots.
\end{gather*}

Analogous to the SK hierarchy, the integrable hierarchy of the KK equation also arises from a generalized bi-Hamiltonian structure, the f\/low is governed by $P_\tau=\bar{G}_l[P]$, where $\bar{G}_l[P]$ are determined by the relations
\begin{gather*}
\bar{G}_l[P]=\bar{\mathcal{L}}\frac{\delta \bar{\mathcal{E}}_{l-1}}{\delta P} \qquad \hbox{and}\qquad \bar{\mathcal{D}}\bar{G}_l[P]=\frac{\delta \bar{\mathcal{E}}_l}{\delta P},\qquad l\in \mathbb{Z},
\end{gather*}
with
\begin{gather}
\begin{split}
& \bar{\mathcal{L}}= -\big( \partial_y^3+2 P \partial_y+2 \partial_y P\big), \\
& \bar{\mathcal{D}}= \partial_y^3+6(P \partial_y+\partial_y P)+4\big(\partial_y^2 P \partial_y^{-1}+\partial_y^{-1} P\partial_y^2\big)+32\big(P^2  \partial_y^{-1}+\partial_y^{-1} P^2\big),\end{split}\label{kk-op-l}
\end{gather}
and $\widehat{\mathcal{R}}=\bar{\mathcal{L}}\bar{\mathcal{D}}$ is the consequent recursion operator. It is easy to see that the KK equation \eqref{kk} in this hierarchy is exactly
\begin{gather*}
P_\tau=\bar{G}_1[P]=\bar{\mathcal{L}}\frac{\delta \bar{\mathcal{E}}_{0}}{\delta P}=\bar{\mathcal{L}}\big(P_{yy}+8P^2\big),
\end{gather*}
with the corresponding Hamiltonian functional
\begin{gather*}
\bar{\mathcal{E}}_0= \int \left( \frac{8}{3}P^3-\frac{1}{2}P_y^2\right) \mathrm{d}y.
\end{gather*}

Similarly, if we use the fact that $\bar{\mathcal{D}}\cdot0=\delta \bar{\mathcal{E}}_{0}/ \delta P$, we may conclude that the negative f\/lows of the KK hierarchy take the form
\begin{gather} \label{kk--n}
\big(\bar{\mathcal{L}}\bar{\mathcal{D}}\big)^l P_\tau=0,\qquad l=1, 2,\ldots.
\end{gather}
It is worth noting that since $\bar{\mathcal{D}}=\partial_y^{-1}\mathcal{A}\partial_y^{-1}$, so the equation \eqref{kk--1'} arising from the compatibility condition of the Lax pair \eqref{iso-kk} and \eqref{laxt-kk} is a reduction of the f\/irst negative f\/low $\bar{\mathcal{L}}\bar{\mathcal{D}}P_\tau=0$.

As before, we hereafter denote, for a positive integer $l$, the $l$-th equation in the positive and negative directions of the DP hierarchy by (DP)$_l$ and (DP)$_{-l}$, respectively, while the $l$-th positive and negative f\/lows of the KK hierarchy by (KK)$_{l}$ and (KK)$_{-l}$, respectively. With this notation, we state the main theorem on the Liouville correspondence between the DP and KK hierarchies as follows.

\begin{Theorem}\label{t3}
Under the Liouville transformations \eqref{recixy-dp} and \eqref{recipn}, for each nonzero integer $0 \ne l \in \mathbb{Z}$, the {\rm(}DP{\rm)}$_{l}$ equation is mapped into the {\rm(}KK{\rm)}$_{-l}$ equation, and conversely.
\end{Theorem}

The proof of this theorem is based on the following two preliminary lemmas, which clarify the relations between certain operators.

\begin{Lemma}\label{l4.1}
Let $n(t, x)$ and $P(\tau, y)$ be related by the transformations \eqref{recixy-dp} and \eqref{recipn}. Then the following identities hold:
\begin{gather}
n^{-\frac{1}{2}} \left(\frac{1}{4}-\partial_x^2\right) n^{-\frac{1}{6}}=-\big(P+\partial_y^2\big), \label{op4}\\
n^{-\frac{2}{3}} \big(\partial_x-\partial_x^3\big) n^{-\frac{1}{3}}=\bar{\mathcal{L}}, \label{op5}\\
n^{-1} \mathcal{D} n^{-\frac{2}{3}}=\partial_y \bar{\mathcal{D}} \partial_y.\label{op6}
\end{gather}
\end{Lemma}

\begin{proof}
{\textbf{(i)}}.\quad Def\/ine $\chi=n^{\frac{1}{3}}$, so from \eqref{recixy-dp} and \eqref{recipn}, we have $\partial_x=\chi \partial_y$ and \begin{gather}\label{recips}
P=\frac{1}{4}\chi^{-2}\chi_y^2-\frac{1}{2}\chi^{-1}\chi_{yy}-\frac{1}{4}\chi^{-2}.
\end{gather}
And then, a direct calculation shows that
\begin{gather*}
\partial_x^2 \chi^{-\frac{1}{2}}=\chi \partial_y \chi \partial_y \chi^{-\frac{1}{2}}=\chi^{\frac{3}{2}} \partial_y^2-\frac{1}{2}\chi \big( \chi^{-\frac{1}{2}} \chi_y\big)_y,
\end{gather*}
where, by \eqref{recips}
\begin{gather*}
\chi \big( \chi^{-\frac{1}{2}} \chi_y\big)_y=-\frac{1}{2} \chi^{-\frac{1}{2}}-2 \chi^{\frac{3}{2}} P.
\end{gather*}
We thus have
\begin{gather*}
\partial_x^2 \chi^{-\frac{1}{2}}=\frac{1}{4} \chi^{-\frac{1}{2}}+\chi^{\frac{3}{2}} \big(P+\partial_y^2\big),
\end{gather*}
which immediately leads to
\begin{gather*}
\left(\frac{1}{4}-\partial_x^2\right) \chi^{-\frac{1}{2}}=-\chi^{\frac{3}{2}} \big(P+\partial_y^2\big),
\end{gather*}
and verif\/ies \eqref{op4}.

{\textbf{(ii)}}.\quad To prove \eqref{op5}, according to \eqref{recixy-dp} and \eqref{recipn}, we only need to verify
\begin{gather}\label{op5'}
\bar{\mathcal{L}}=\chi^{-1} \partial_y \big(1-\partial_x^2\big) \chi^{-1}.
\end{gather}
In fact, since \begin{gather*}\big(1-\partial_x^2\big) \chi^{-1}=\chi^{-1}+\chi\big(\chi^{-1}\chi_y\big)_y+\chi_y\partial_y-\chi\partial_y^2,\end{gather*} one has
\begin{gather*}
\chi^{-1} \partial_y \big(1-\partial_x^2\big) \chi^{-1}=\chi^{-1} \big( \chi^{-1}+\chi \big(\chi^{-1} \chi_y\big)_y\big)_y+\chi^{-1} \big( \chi^{-1} +\chi \big(\chi^{-1} \chi_y\big)_y+\chi_{yy}\big) \partial_y-\partial_y^3.
\end{gather*}
This, when combined with \eqref{recips}, proves \eqref{op5'}.

{\textbf{(iii)}}.\quad In view of the explicit form of $\mathcal{D}$, we deduce that
\begin{gather*}
\chi^{-3} \mathcal{D} \chi^{-2}=\chi^{-1} \bar{\mathcal{L}} \chi \big(4-\partial_x^2\big) \chi^{-2},
\end{gather*}
where we have made used of \eqref{op5}. In the right-hand side of the preceding equation
 \begin{gather*}
\chi \big(4-\partial_x^2\big)\chi^{-2} = \chi\big( 4\chi^{-2}+2\chi \big(\chi^{-2} \chi_y\big)_y+3\chi^{-1} \chi_y\partial_y-\partial_y^2\big)\\
 \hphantom{\chi \big(4-\partial_x^2\big)\chi^{-2}}{} = -16\chi P-6\chi_{yy}+3\chi_y\partial_y-\chi\partial_y^2.
\end{gather*}
Hence,
\begin{gather*}
\chi^{-1} \bar{\mathcal{L}} \chi \big(4-\partial_x^2\big) \chi^{-2}\\
\qquad{} =\chi^{-1}\big( 32\chi PP_y+64(\chi P)_yP+12\chi_{yy}P_y+24\chi_{yyy}P+16(\chi P)_{yyy}+6\chi_{yyyyy}\big)\\
\qquad \quad{} +\chi^{-1}\big( 64\chi P^2-6\chi_yP_y+12\chi_{yy}P+48(\chi P)_{yy}+15\chi_{yyyy}\big)\partial_y\\
\qquad \quad{} +\chi^{-1}\big( 2\chi P_y-8\chi_yP+48(\chi P)_y+10\chi_{yyy}\big)\partial_y^2+20P\partial_y^3+\partial_y^5\\
\qquad{} =4 ( 16PP_y+P_{yyy} )+2\big( 32P^2+9P_{yy}\big) \partial_y+30P_y\partial_y^2+20P\partial_y^3+\partial_y^5,
 \end{gather*}
which implies
\begin{gather*}\chi^{-3} \mathcal{D} \chi^{-2}=\partial_y \bar{\mathcal{D}} \partial_y,\end{gather*}
and then \eqref{op6} follows.
\end{proof}

\begin{Lemma}\label{lem3.2}
Under the transformations \eqref{recixy-dp} and \eqref{recipn}, the relation
\begin{gather}\label{l3.2}
n^{-1}\big(\mathcal{D} \mathcal{L}^{-1}\big)^l n=\partial_y\big(\bar{\mathcal{D}}\bar{\mathcal{L}}\big)^l\partial_y^{-1}
\end{gather}
holds for each integer $l\geq 1$.
\end{Lemma}

\begin{proof}
Due to the form of the inverse operator $\mathcal{L}^{-1}$ and the identities \eqref{op5} and \eqref{op6}, we arrive at
\begin{gather*}
n^{-1} \mathcal{D}\mathcal{L}^{-1} n =\partial_y \bar{\mathcal{D}} \partial_y \partial_x^{-1} n^{-\frac{1}{3}} \big(\partial_x-\partial_x^3\big) n^{-\frac{1}{3}} \partial_x^{-1} n^{\frac{1}{3}}=\partial_y \bar{\mathcal{D}} \bar{\mathcal{L}} \partial_y^{-1},
\end{gather*}
which verif\/ies \eqref{l3.2} for $l=1$. Then an obvious induction procedure allows us to prove \eqref{l3.2} in general. Hence the lemma is proved.
\end{proof}

\begin{proof}[Proof of Theorem \ref{t3}]
To prove this theorem, we take the analogous steps as in the proof of Theorem \ref{t1}. First of all, the derivative of $P$ with respect to $t$ is
\begin{gather*}
P_t=P_\tau+P_y\int^x \big( n^{\frac{1}{3}}(t, \xi)\big)_t \mathrm{d}\xi
=P_\tau+\frac{1}{3}P_y\partial_x^{-1}n^{-\frac{2}{3}}n_t=P_\tau+\frac{1}{3}P_y\partial_y^{-1}n^{-1}n_t.
\end{gather*}
On the other hand, it follows from \eqref{recipn} that
\begin{gather*}
P_t =\frac{1}{8}n^{-\frac{3}{2}}n_t \big(1-4\partial_x^2\big) n^{-\frac{1}{6}}+\frac{1}{24}n^{-\frac{1}{2}} \big(1-4\partial_x^2\big) n^{-\frac{7}{6}} n_t\\
\hphantom{P_t}{} =-\frac{1}{2}Pn^{-1}n_t+\frac{1}{24}n^{-\frac{1}{2}} \big(1-4\partial_x^2\big) n^{-\frac{7}{6}} n_t.
\end{gather*}
From the preceding equations and using the formula \eqref{op4}, we have
\begin{gather}\label{t4-1}
P_\tau=\left( -\frac{1}{3} P_y\partial_y^{-1}-\frac{1}{2}P-\frac{1}{6}\big(P+\partial_y^2\big)\right) n^{-1}n_t=\frac{1}{6}\bar{\mathcal{L}} \partial_y^{-1}n^{-1}n_t.
\end{gather}

Now, suppose $n(t, x)$ is the solution of the (DP)$_{-l}$ equation
 \begin{gather}\label{dp--l}
n_t=G_{-l}=6\big( \mathcal{D} \mathcal{L}^{-1}\big) ^{l-1}\mathcal{D}n^{-\frac{2}{3}},\qquad l=1, 2,\ldots,
\end{gather}
and $n(t, x)$ is related to $P(\tau, y)$ according to \eqref{recixy-dp} and \eqref{recipn}. Then plugging \eqref{dp--l} into \eqref{t4-1}, one f\/inds that the corresponding function $P(\tau, y)$ satisf\/ies
\begin{gather*}
P_\tau =\bar{\mathcal{L}}\partial_y^{-1}n^{-1}\big( \mathcal{D} \mathcal{L}^{-1}\big) ^{l-1}\mathcal{D}n^{-\frac{2}{3}}=\bar{\mathcal{L}}\big( \bar{\mathcal{D}} \bar{\mathcal{L}}\big) ^{l-1}\partial_y^{-1}n^{-1}\mathcal{D}n^{-\frac{2}{3}}\\
 \hphantom{P_\tau}{} =\big( \bar{\mathcal{L}} \bar{\mathcal{D}}\big) ^{l-1} \bar{\mathcal{L}} \bar{\mathcal{D}}\cdot 0=\big( \bar{\mathcal{L}} \bar{\mathcal{D}}\big) ^{l-1} \bar{\mathcal{L}} \frac{\delta \bar{\mathcal{E}}_0}{\delta P}=\bar{G}_l.
\end{gather*}
This shows that by the transformations \eqref{recixy-dp} and \eqref{recipn}, the (DP)$_{-l}$ equation is mapped into the (KK)$_{l}$ equation.

When it comes to the (DP)$_{l}$ equation for $l\geq 1$, inserting the formula \begin{gather*}
n_t=G_{l}=\big( \mathcal{L} \mathcal{D}^{-1}\big) ^{l-1}\mathcal{L} \frac{\delta \mathcal{E}_0}{\delta n}=\big( \mathcal{L} \mathcal{D}^{-1}\big) ^{l-1}\mathcal{L}\cdot \frac{9}{2}
\end{gather*}
for the (DP)$_{l}$ equation
into \eqref{t4-1} yields
\begin{gather*}
P_\tau=\frac{3}{4}\bar{\mathcal{L}} \partial_y^{-1}n^{-1} \big( \mathcal{L} \mathcal{D}^{-1}\big) ^{l-1} \mathcal{L} \cdot 1.
\end{gather*}
Therefore, referring back to the form of the operator $\mathcal{D}$ and using the identity \eqref{l3.2}, we deduce that
\begin{gather*}
\big(\bar{\mathcal{L}} \bar{\mathcal{D}}\big) ^{l}P_\tau
 =\frac{3}{4}\bar{\mathcal{L}} \big(\bar{\mathcal{D}} \bar{\mathcal{L}}\big)^{l} \partial_y^{-1} n^{-1} \big( \mathcal{L} \mathcal{D}^{-1}\big) ^{l-1} \mathcal{L}\cdot 1\\
\hphantom{\big(\bar{\mathcal{L}} \bar{\mathcal{D}}\big) ^{l}P_\tau}{}
=\frac{3}{4}\bar{\mathcal{L}} \partial_y^{-1} n^{-1} \mathcal{D} \cdot 1=-\frac{3}{4}\big(\partial_y^2+4P+2P_y\partial_y^{-1}\big) \cdot 0=0,
\end{gather*}
which shows that $P(\tau, y)$ is a solution for the (KK)$_{-l}$ equation \eqref{kk--n}. We thus have proved that for each $l\geq 1$, the (DP)$_{l}$ equation is mapped, via the transformations \eqref{recixy-dp} and \eqref{recipn}, into the (KK)$_{-l}$ equation.

As in Theorem \ref{t1}, the converse argument is also valid.
\end{proof}

\subsection[The correspondence between the Hamiltonian conservation laws of the DP and KK equations]{The correspondence between the Hamiltonian conservation laws\\ of the DP and KK equations}\label{section3.3}

We now investigate the relationship between the Hamiltonian conservation laws of the DP and KK equations. For the DP equation \eqref{dp},
with the Hamiltonian pair $\mathcal{L}$ and $\mathcal{D}$ def\/ined in \eqref{dp-ld} in hand, the corresponding recursive formula formally def\/ines an inf\/inite hierarchy of Hamiltonian conservation laws $\mathcal{E}_{l}$ determined by
\begin{gather}\label{dpclhie}
\mathcal{L} \frac{\delta \mathcal{E}_{l-1}}{\delta n}=\mathcal{D} \frac{\delta \mathcal{E}_l}{\delta n},\qquad l\in \mathbb{Z}.
\end{gather}
For the KK equation, its Hamiltonian conservation laws $\bar{\mathcal{E}}_{l}$ can be determined by the generalized bi-Hamiltonian (bi-Dirac) structure
\begin{gather*}
\bar{\mathcal{D}} \bar{\mathcal{L}} \frac{\delta \bar{\mathcal{E}}_{l-1}}{\delta P}= \frac{\delta \bar{\mathcal{E}}_l}{\delta P},\qquad l\in \mathbb{Z},
\end{gather*}
with $\bar{\mathcal{L}}$ and $\bar{\mathcal{D}}$ given by \eqref{kk-op-l}.

Before proving the main theorem for the correspondence between the two hierarchies of Hamiltonian conservation laws $\{\mathcal{E}_l\}$ and $\{\bar{\mathcal{E}}_l\}$, two lemmas regarding their variational derivatives are in order.

\begin{Lemma}
Let $\{\mathcal{E}_l\}$ and $\{\bar{\mathcal{E}}_l\}$ be the hierarchies of Hamiltonian conservation laws of the DP and KK equations, respectively. Then, for each $l\in \mathbb{Z}$, their corresponding variational derivatives are related according to the following identity
\begin{gather}\label{l4.3-1}
\frac{\delta \mathcal{E}_{l}}{\delta n}=6\mathcal{L}^{-1} n \partial_y \frac{\delta \bar{ \mathcal{E}}_{-(l+2)}}{\delta P}.
\end{gather}
\end{Lemma}

\begin{proof}
We f\/irst consider the case of $l\leq -2$. Since
\begin{gather*}
\frac{\delta \mathcal{E}_{-2}}{\delta n}=\mathcal{L}^{-1}\mathcal{D}\frac{\delta \mathcal{E}_{-1}}{\delta n}=6\mathcal{L}^{-1}\mathcal{D} n^{-\frac{2}{3}}
=6\mathcal{L}^{-1} n \partial_y \bar{\mathcal{D}}\cdot 0,
\end{gather*}
then the fact $\bar{\mathcal{D}} \cdot 0=\delta \bar{\mathcal{E}}_{0}/\delta P$ reveals that \eqref{l4.3-1} holds for $l=-2$.

We proceed by induction on $l$. Assume that \eqref{l4.3-1} holds when $l=k$, namely
\begin{gather*}
\frac{\delta \mathcal{E}_{k}}{\delta n}=6\mathcal{L}^{-1} n \partial_y \frac{\delta \bar{ \mathcal{E}}_{-(k+2)}}{\delta P}.
\end{gather*}
From the recursive formula \eqref{dpclhie} and by the assumption,
\begin{gather*}
\frac{\delta \bar{\mathcal{E}}_{-(k+1)}}{\delta P} =\bar{\mathcal{D}} \bar{\mathcal{L}} \frac{\delta \bar{\mathcal{E}}_{-(k+2)}}{\delta P}=\frac{1}{6}\bar{\mathcal{D}} \bar{\mathcal{L}} \partial_y^{-1} n^{-1} \mathcal{L}\frac{\delta \mathcal{E}_{k}}{\delta n}
=\frac{1}{6}\bar{\mathcal{D}} \bar{\mathcal{L}} \partial_y^{-1} n^{-1} \mathcal{L} \mathcal{D}^{-1} \mathcal{L} \frac{\delta \mathcal{E}_{k-1}}{\delta n}.
\end{gather*}
Then, thanks to Lemma \ref{lem3.2} with $l=1$, we conclude that \eqref{l4.3-1} holds for $l=k-1$.

Furthermore, for the case of $l=-1$, we claim
\begin{gather*}
\frac{\delta \bar{\mathcal{E}}_{-1}}{\delta P}=\frac{1}{6}\partial_y^{-1} n^{-1} \mathcal{L} \frac{\delta \mathcal{E}_{-1}}{\delta n}.
\end{gather*}
Note that $\bar{\mathcal{E}}_{-1}$ is a Casimir functional for Hamiltonian operator $\bar{\mathcal{L}}$, it suf\/f\/ices to show that
\begin{gather*}
\bar{\mathcal{L}} \partial_y^{-1} n^{-1} \mathcal{L} \frac{\delta \mathcal{E}_{-1}}{\delta n}=0.
\end{gather*}
Indeed, from the def\/inition of the operator $\mathcal{L}$ and the formula \eqref{op5}, we have
\begin{gather*}
\bar{\mathcal{L}} \partial_y^{-1} n^{-1} \mathcal{L} n^{-\frac{2}{3}}= n^{-\frac{1}{3}} \partial_x\cdot 1=0,
\end{gather*}
proving the claim and verifying that \eqref{l4.3-1} holds for $l=-1$.

Finally, induction on $l$ shows that if \eqref{l4.3-1} holds for $l=k$, then for $l=k+1$, from the recursive formula \eqref{dpclhie} and the identities \eqref{op5} and \eqref{op6}, we infer that
\begin{gather*}
\frac{\delta \mathcal{E}_{k+1}}{\delta n} =\mathcal{D}^{-1}\mathcal{L} \frac{\delta \mathcal{E}_{k}}{\delta n}=6\mathcal{D}^{-1} n \partial_y \frac{\delta \bar{\mathcal{E}}_{-(k+2)}}{\delta P}=6\mathcal{D}^{-1} n \partial_y \bar{\mathcal{D}} \bar{\mathcal{L}} \frac{\delta \bar{\mathcal{E}}_{-(k+3)}}{\delta P}\\
\hphantom{\frac{\delta \mathcal{E}_{k+1}}{\delta n}}{}
 =6n^{-\frac{2}{3}} \partial_x^{-1} n^{-\frac{1}{3}} \big(\partial_x-\partial_x^3\big) n^{-\frac{1}{3}} \frac{\delta \bar{\mathcal{E}}_{-(k+3)}}{\delta P}=6\mathcal{L}^{-1} n \partial_y \frac{\delta \bar{\mathcal{E}}_{-(k+3)}}{\delta P},
\end{gather*}
which completes the induction step, and thereby proves the lemma.
\end {proof}

\begin{Lemma} \label{lem3.4}
Let $n(t, x)$ and $P(\tau, y)$ be related by the transformations \eqref{recixy-dp} and \eqref{recipn}. If
$\mathcal{E}(n)=\bar{\mathcal{E}}(P)$, then
\begin{gather}\label{vdpn}
\frac{\delta \mathcal{E}}{\delta n}=\frac{1}{6} n^{-\frac{2}{3}}\partial_y^{-1}\bar{\mathcal{L}} \frac{\delta \bar{\mathcal{E}}}{\delta P},
\end{gather}
where $\bar{\mathcal{L}}$ is the Hamiltonian operator \eqref{kk-op-l} admitted by the KK equation \eqref{kk}.
\end{Lemma}

\begin{proof}
As the f\/irst step, in view of \eqref{recipn}, for convenience, we introduce the notation
\begin{gather*}
\tilde{F}[n(t, x)]\equiv -n^{-\frac{1}{2}}\left( \frac{1}{4}-\partial_x^2 \right) n^{-\frac{1}{6}}=P(\tau, y).
\end{gather*}
Evaluating the Fr\'echet derivative of $\tilde{F}[n]$ produces
\begin{gather*}
\mathcal{D}_{\tilde{F}[n]} =\frac{1}{6}n^{-\frac{5}{3}}-\frac{1}{2}n^{-\frac{3}{2}}\big( n^{-\frac{1}{6}}\big)_{xx}-\frac{1}{6}n^{-\frac{1}{2}}\big( n^{-\frac{7}{6}}\big)_{xx}\\
\hphantom{\mathcal{D}_{\tilde{F}[n]}}{}
=-\frac{1}{2}n^{-1} P-\frac{1}{6} \big(P+\partial_y^2\big) n^{-1}=-\frac{1}{6}\big(4P+\partial_y^2\big) n^{-1}.
\end{gather*}
On the other hand, we get
\begin{gather*}
\frac{\mathrm{d}}{\mathrm{d}\epsilon}\Big|_{\epsilon=0}^{y \, {\rm f\/ixed}}\tilde{F}[n+\epsilon \rho]=\mathcal{D}_{\tilde{F}}(\rho)-\frac{1}{3}P_y \partial_y^{-1} n^{-1}\rho=\frac{1}{6}\bar{\mathcal{L}} \partial_y^{-1} n^{-1}\rho.
\end{gather*}

Secondly, by the assumption, one has
\begin{gather*}
\frac{\mathrm{d}}{\mathrm{d}\epsilon}\Big|_{\epsilon=0} \mathcal{E}(n+\epsilon \rho)=\frac{\mathrm{d}}{\mathrm{d}\epsilon}\Big|_{\epsilon=0} \bar{\mathcal{E}}\big(\tilde{F}[n+\epsilon \rho]\big).
\end{gather*}
Furthermore, due to the fact that $\bar{\mathcal{L}}$ is skew-adjoint, the righ-hand side of the above expression yields
\begin{gather*}
\frac{\mathrm{d}}{\mathrm{d}\epsilon}\Big|_{\epsilon=0} \bar{\mathcal{E}}\big(\tilde{F}[n+\epsilon \rho]\big)
=\int \frac{\delta \bar{\mathcal{E}}}{\delta P} \cdot\frac{\mathrm{d}}{\mathrm{d}\epsilon}\Big|_{\epsilon=0}^{y \, {\rm f\/ixed}}\tilde{F}[n+\epsilon \rho] \mathrm{d}y\\
\hphantom{\frac{\mathrm{d}}{\mathrm{d}\epsilon}\Big|_{\epsilon=0} \bar{\mathcal{E}}\big(\tilde{F}[n+\epsilon \rho]\big)}{}
 =\int \frac{\delta \bar{\mathcal{E}}}{\delta P} \cdot\left(\frac{1}{6}\bar{\mathcal{L}} \partial_y^{-1} n^{-1}\rho\right) \mathrm{d}y
=\frac{1}{6}\int n^{-\frac{2}{3}} \partial_y^{-1} \bar{\mathcal{L}} \frac{\delta \bar{\mathcal{E}}}{\delta P} \rho \mathrm{d}x,
\end{gather*}
which, combined with the def\/inition of the variational derivative, produces \eqref{vdpn}.
\end{proof}

Finally, it follows from \eqref{op6} and \eqref{l3.2} with $l=1$ that
\begin{gather*}
\mathcal{L}^{-1}n \partial_y=\mathcal{D}^{-1}n \partial_y \bar{\mathcal{D}} \bar{\mathcal{L}}=n^{-\frac{2}{3}} \partial_y^{-1} \bar{\mathcal{L}},
\end{gather*}
which, together with \eqref{l4.3-1}, implies
\begin{gather}\label{l4.3-1'}
\frac{\delta \mathcal{E}_{l}}{\delta n}=6n^{-\frac{2}{3}}\partial_y^{-1}\bar{\mathcal{L}}\frac{\delta \bar{ \mathcal{E}}_{-(l+2)}}{\delta P}.
\end{gather}
Now, we suppose that $n(t, x)$ and $P(\tau, y)$ are related by the transformations \eqref{recixy-dp} and \eqref{recipn}. Def\/ine a functional
\begin{gather*}
\label{t3.2-1}\tilde{\mathcal{G}}_{k}(P)\equiv \mathcal{E}_{l}(n),
\end{gather*}
for some $k \in \mathbb{Z}$. Then from Lemma \ref{lem3.4}, we conclude, for each $k\in \mathbb{Z}$,
\begin{gather*}
\frac{\delta \mathcal{E}_{l}}{\delta n}= \frac{1}{6}n^{-\frac{2}{3}} \partial_y^{-1}\bar{\mathcal{L}}\frac{\delta \tilde{\mathcal{G}}_{k}}{\delta P},
\end{gather*}
which, in comparision with \eqref{l4.3-1'} produces
\begin{gather*}
\mathcal{E}_l(n)=\tilde{\mathcal{G}}_{k}(P)=36 \bar{\mathcal{E}}_{-(l+2)}(P).
\end{gather*}
As a consequence, the following theorem is thereby proved.

\begin{Theorem}\label{t4}
Under the Liouville transformations \eqref{recixy-dp} and \eqref{recipn}, for each $l\in \mathbb{Z}$, the Hamiltonian conservation law $\bar{\mathcal{E}}_l(P)$ of the KK equation is related to that $\mathcal{E}_{l}(n)$ of the DP equation, according to the following identity
\begin{gather*}
\mathcal{E}_{l}(n)=36 \bar{\mathcal{E}}_{-(l+2)}(P), \qquad l\in \mathbb{Z}.
\end{gather*}
\end{Theorem}

For example, in the case of $l=-2$, it is inferred from \eqref{op5} and \eqref{op6}, and using the fact $\bar{\mathcal{D}}\cdot 0=\delta \bar{\mathcal{E}}_0/ \delta P=(P_{yy}+8P^2)$, that
\begin{gather*}
\frac{\delta \mathcal{E}_{-2}}{\delta n}=\mathcal{L}^{-1}\mathcal{D}\frac{\delta \mathcal{E}_{-1}}{\delta n}
=6n^{-\frac{2}{3}} \partial_y^{-1} \bar{\mathcal{L}} \bar{\mathcal{D}}\cdot 0
=-6n^{-\frac{2}{3}} \left( P_{yyyy}+20PP_{yy}+15P_y^2+\frac{80}{3}P^3\right),
\end{gather*}
with $P$ satisfying \eqref{recipn}. Consequently, we arrive at
\begin{gather*}
\mathcal{E}_{-2}(n) =\frac{18}{5}\int n^{\frac{1}{3}}\left( P_{yyyy}+20PP_{yy}+15P_y^2+\frac{80}{3}P^3\right) \mathrm{d}x\\
\hphantom{\mathcal{E}_{-2}(n)}{} =36\int \left( \frac{8}{3}P^3-\frac{1}{2}P_y^2\right) \mathrm{d}y=36 \bar{\mathcal{E}}_{0}(P),
\end{gather*}
which is in accordance with Theorem \ref{t4}.

\section[The relationship between the Novikov equation and the DP equation]{The relationship between the Novikov equation\\ and the DP equation}\label{section4}

It has been shown in \cite{fg} that under the Miura transformations
\begin{gather}\label{mitr-sk}
\mathcal{B}_1(Q, V)\equiv Q-V_y+V^2=0
\end{gather}
and \begin{gather}\label{mitr-kk}
\mathcal{B}_2(P, V)\equiv P+V_y+\frac{1}{2}V^2=0,
\end{gather}
the SK equation \eqref{sk} and the KK equation \eqref{kk} are respectively transformed into the \emph{Fordy--Gibbons--Jimbo--Miwa equation}
\begin{gather}\label{mitr-eq}
V_\tau+V_{yyyyy}-5\big(V_yV_{yyy}+V_{yy}^2+V_y^3+4VV_yV_{yy}+V^2V_{yyy}-V^4V_y\big)=0.
\end{gather}
This, together with the fact that there exist the Liouville correspondences between the Novikov and SK hierarchies, as well as between the DP and KK hierarchies, inspires a natural question as to whether there exists some relationship between the Novikov equation~\eqref{nov} and the DP equation~\eqref{dp}.

We can regard \eqref{mitr-sk} and \eqref{mitr-kk} as B\"acklund transformations. According to Fokas and Fuchs\-stei\-ner~\cite{ff1}, all the positive f\/lows in the SK hierarchy admit the same transformation \eqref{mitr-sk}. More precisely, set
\begin{gather}\label{T1}
T_1\equiv \mathcal{B}_{1,V}^{-1}\mathcal{B}_{1,Q}=( 2V-\partial_y)^{-1},
\end{gather}
where $\mathcal{B}_{1,V}$ and $\mathcal{B}_{1,Q}$ are the Fr\'echet derivatives of \eqref{mitr-sk} with respect to $V$ and $Q$, respectively. Then, the recursion operator $\bar{\mathcal{R}}$ of the SK equation and the recursion operator~$\mathcal{R}^*$ of equation~\eqref{mitr-eq} satisfy
\begin{gather*}
\mathcal{R}^*=T_1\bar{\mathcal{R}}T_1^{-1} .
\end{gather*}
Similarly, each member in the KK hierarchy admits the Miura transformation~\eqref{mitr-kk}, and its corresponding recursion operator $\widehat{\mathcal{R}}$ is linked with the recursion operator $\mathcal{R}^*$ according to the identity
\begin{gather*}
\mathcal{R}^*=T_2\widehat{\mathcal{R}} T_2^{-1} ,
\end{gather*}
where
\begin{gather*}
T_2\equiv \mathcal{B}_{2,V}^{-1}\mathcal{B}_{2,P}= ( V+\partial_y )^{-1}
\end{gather*}
is the operator arising from the function \eqref{mitr-kk}.

In light of these relations, we claim that both the f\/irst negative f\/low of the SK hierarchy and the KK hierarchy are related to the same equation
\begin{gather}\label{mitr-eq-1}
\mathcal{R}^* V_{\tau}=0,
\end{gather}
via the Miura transformations \eqref{mitr-sk} and \eqref{mitr-kk}, respectively. Indeed, we have the following result.

\begin{Proposition}\label{p4.1}
Assume that $V$ satisfies the equation \eqref{mitr-eq-1}.
Then $Q=V_y-V^2$ and $P=-V_y-\frac{1}{2} V^2$ satisfy the first negative flow of the SK hierarchy $\bar{\mathcal{R}} Q_{\tau}=0$ and the first negative flow of the KK hierarchy $\widehat{\mathcal{R}} P_{\tau}=0$, respectively.
\end{Proposition}

\begin{proof}
Thanks to \eqref{mitr-sk}, one has
\begin{gather*}\label{3.3.3}
Q_{\tau}=V_{y\tau}-2VV_{\tau}=-T_1^{-1}V_{\tau}.
\end{gather*}
This, together with \eqref{T1}, implies
\begin{gather*}
\bar{\mathcal{R}} Q_{\tau}=-\bar{\mathcal{R}} T_1^{-1} V_{\tau}=-T_1^{-1} \mathcal{R}^* V_{\tau}.
\end{gather*}
Therefore, if $\mathcal{R}^* V_{\tau}=0$, then $\bar{\mathcal{R}} Q_{\tau}=-(2V-\partial_y) \mathcal{R}^* V_{\tau}=0$, proving the SK part of the proposition. The KK part can be proved by in a similar manner.
\end{proof}

Finally, using Proposition \ref{p4.1}, combined with Theorems \ref{t1} and \ref{t3}, we are able to establish a relationship between the Novikov equation \eqref{nov} and the DP equation \eqref{dp}. This fact is summarized in the following proposition.

\begin{Proposition}\label{p4.2}
Both the Novikov equation \eqref{nov} and the DP equation \eqref{dp} are linked with the equation \eqref{mitr-eq-1}
in the following sense. If $V(\tau, y)$ is a solution of equation \eqref{mitr-eq-1}, then the function $m(t, x)$ determined implicitly by the relation
\begin{gather*}
V_y-V^2=-m^{-1} \big(1-\partial_x^2\big) m^{-\frac{1}{3}},\qquad y=\int^x m^{\frac{2}{3}}(t, \xi) \mathrm{d}\xi, \qquad \tau=t,
\end{gather*}
satisfies the Novikov equation \eqref{nov}, while the function $n(t, x)$ determined by
\begin{gather*}
V_y+\frac{1}{2}V^2=\frac{1}{4}n^{-\frac{1}{2}}\big(1-4\partial_x^2\big) n^{-\frac{1}{6}},\qquad y=\int^x n^{\frac{1}{3}}(t, \xi) \mathrm{d}\xi, \qquad \tau=t,
\end{gather*}
satisfies the DP equation \eqref{dp}.
\end{Proposition}

\subsection*{Acknowledgements}

The authors thank the referees for valuable comments and suggestions. Kang's research was supported by NSFC under Grant 11631007 and Grant 11471260. Liu's research was supported in part by NSFC under Grant 11631007 and Grant 11401471, and Ph.D.~Programs Foundation of Ministry of Education of China-20136101120017. Olver's research was supported by NSF under Grant DMS-1108894. Qu's research was supported by NSFC under Grant 11631007 and Grant 11471174.

\pdfbookmark[1]{References}{ref}
\LastPageEnding

\end{document}